\documentclass[a4paper,11pt,oneside]{article}
\usepackage[a4paper, margin=1in]{geometry}
\usepackage{setspace}
\usepackage{textcomp}
\usepackage{enumerate}
\usepackage{amsmath,amssymb,amsfonts,amsthm,mathrsfs}
\usepackage{courier}
\usepackage{algpseudocode,algorithm}
\usepackage{cancel, cases}
\usepackage{soul}
\usepackage{float,graphicx,subcaption,epstopdf}
\usepackage{longtable,tabu,threeparttable,rotating,adjustbox}
\usepackage[normalem]{ulem}
\usepackage[dvipsnames]{xcolor}
\usepackage{tcolorbox} 

\usepackage{hyperref}

\usepackage{caption}

\usepackage[round, authoryear]{natbib}

\DeclareFontFamily{OT1}{pzc}{} \DeclareFontShape{OT1}{pzc}{m}{it}%
{<-> s * [0.900] pzcmi7t}{} \DeclareMathAlphabet{\mathpzc}{OT1}{pzc}%
{m}{it}

\newcommand\tr{\text{tr}}

\newcommand{\bi}{\begin{itemize}}
\newcommand{\ei}{\end{itemize}}

\theoremstyle{plain}
\newtheorem{theorem}{Theorem}
\newtheorem{corollary}[theorem]{Corollary}
\newtheorem{lemma}[theorem]{Lemma}
\newtheorem{proposition}[theorem]{Proposition}

\theoremstyle{definition}
\newtheorem{definition}[theorem]{Definition}
\newtheorem{assumption}[theorem]{Assumption}

\theoremstyle{remark}
\newtheorem{remark}{Remark}

\numberwithin{equation}{section}
\numberwithin{theorem}{section}

\title{Risk-Sensitive Investment Management via Free Energy--Entropy Duality}
\author{S\'ebastien Lleo\footnote{inance Department and `AI, Data Science \& Business' AE, NEOMA Business School, France} \and Wolfgang Runggaldier\footnote{University of Padova, Italy, and Fellow, Institut Louis Bachelier, Paris, France}}
\date{\today}

\begin{document}
\maketitle

\begin{abstract}
We study a benchmarked risk-sensitive portfolio problem in a factor-based setting to bring together three strands of the literature: benchmarked risk-sensitive investment management, the Kuroda–Nagai change-of-measure method, and the free energy–entropy duality of \citet{daipraConnectionsStochasticControl1996}. We show that the duality yields a direct solution of the benchmarked problem by reformulating it as a linear-quadratic-Gaussian stochastic differential game under a suitable equivalent probability measure, with an entropic regularization. The resulting value function is quadratic, the optimal controls are explicit affine feedback maps, and the optimal allocation admits two complementary interpretations: as a fractional Kelly strategy and as a Kelly portfolio adjusted via the entropic regularization. This formulation, therefore, contributes both a direct analytical route to the solution and a clearer interpretation of risk sensitivity, thereby embedding the classical Kuroda–Nagai change-of-measure approach within a more general framework.  An added benefit of this formulation is that it is suitable for implementation via an RL algorithm. A simple implementation on U.S. equity data illustrates the tractability of the framework and numerically confirms the equivalence of the two approaches.
\end{abstract}



\textbf{MSC classes:} 91G10; 91G80; 93E20.

\section{Introduction}\label{sec:Intro}

This paper brings together three strands of literature in a single tractable framework: benchmarked risk-sensitive investment management, the Kuroda--Nagai change-of-measure method, and the free energy--relative entropy duality of \citet{daipraConnectionsStochasticControl1996}.

Over the past quarter-century, risk-sensitive control has become a well-established approach to dynamic portfolio selection. One reason for its appeal in financial applications is that it admits several useful interpretations: it connects naturally to utility theory, provides a dynamic counterpart to mean--variance optimization, converges to the Kelly criterion in the limit as the risk sensitivity parameter $\theta \to 0$, and, for $\theta>0$, the associated objective relates closely to the entropic risk measure \citep{Bielecki2003, lleoDualDominanceHow2025}. At the same time, risk-sensitive investment problems are not standard control problems. In the benchmarked setting, the state process $X_s$ is uncontrolled, while the terminal reward depends on a controlled, stochastic rate-of-return process that contains an It\^o integral. Because the criterion is exponential, this It\^o integral does not disappear under expectation. For some time, this feature made such problems difficult to solve directly \citep[see][]{bipl99}. Kuroda and Nagai later showed that a suitable change of measure recasts the risk-sensitive investment criterion into a standard linear exponential-of-quadratic Gaussian (LEQG) control problem \citep{kuna02}, which can then be solved explicitly; see also \citet{Bensoussan1992}. Together, the factor-based framework of Bielecki and Pliska and the Kuroda--Nagai change-of-measure method shaped a large subsequent literature, including models of factor-based asset management \citep{bipl99,bipl00,hataRisksensitiveAssetManagement2017a}, insider trading \citep{hataRisksensitivePortfolioOptimization2018}, benchmarked asset management \citep{dall_RSBench}, asset--liability management \citep[Chapter 4]{DavisLleoBook2014}, stochastic interest rates \citep{bipl04,hataRiskSensitiveAssetManagement2021}, jump-diffusion dynamics \citep{dall_JDRSAM_Diff,dall_JDRSAM,dall_JDBench,dall_JDRSALM}, partial observation \citet{nape02, hataRisksensitiveStochasticControl2006, LleoRunggaldierSeparation24}, and, more recently, expert opinions \citep{dall_BLcontinuous,davisRisksensitiveBenchmarkedAsset2021,dall_JDBenchAltData2024}. In parallel, risk-sensitive control also saw methodological developments, particularly in the context of model uncertainty \citet{bieleckiRiskSensitiveMarkovDecision2022a}.

The benchmarked formulation we adopt here is economically natural in delegated portfolio management, where investors are often evaluated relative to a reference portfolio rather than in absolute terms. It also substantially broadens the scope of the framework: by varying the benchmark, one can represent a continuum of investment styles, ranging from strongly active mandates to tightly benchmarked or index-tracking mandates, while retaining the same underlying factor structure. At the same time, this generalization comes at essentially no additional mathematical cost, since it amounts to replacing log wealth by the log excess return relative to the benchmark. In particular, the traditional asset-only formulation is recovered as a special case by setting the benchmark to the investor's initial wealth.

This paper makes two main contributions. First, it provides a direct solution to the benchmarked risk-sensitive portfolio problem by leveraging the free energy--entropy duality and an equivalent change of probability measure \citep[see][]{daipraConnectionsStochasticControl1996}. Second, we rigorously demonstrate that this novel solution technique is fully equivalent to the classical Kuroda–Nagai approach. Overall, the advantages of the free energy–entropy duality formulation are twofold: on the one hand, it grounds the solution method in solid theoretical principles; on the other hand, it enhances the transparency of the resulting allocation structure. Notably, the optimal portfolio can now be interpreted as a fractional Kelly strategy, or equivalently, as a Kelly portfolio corrected by a term induced by the duality. Thus, our approach not only introduces a new analytical pathway to the solution but also sharpens the understanding of risk-sensitive portfolio choice as optimization under a penalized adverse probability tilt. In a companion paper, \citet{LleoRunggaldier_EntropyRegularizationinRLandRSIM_2026} relax the known-parameters assumption and demonstrate how this framework extends naturally to reinforcement learning settings.

Section 2 develops the analytical core of the paper. Using free energy--entropy duality, we reformulate the original benchmarked problem as a linear-quadratic-Gaussian (LQG) stochastic differential game under a
suitable equivalent probability measure,  with an entropic
regularization arising from the change of measure. This reformulation makes the problem tractable: the optimal controls are explicit, affine feedback rules, the value function is quadratic, and its coefficients solve a Riccati equation together with related linear equations. It also yields structural insight. The optimal strategy can be decomposed into a fractional Kelly allocation, a benchmark-tracking portfolio, and an intertemporal hedging portfolio, or equivalently represented as a Kelly allocation corrected by a term induced by the duality. In this way, the dual formulation does not merely solve the problem; it also interprets risk sensitivity through the lens of robustness to adverse probability tilts, with entropy regularization shaping the control.

Section 3 shows that free energy--entropy duality offers a direct route to the classical solution of the benchmarked risk-sensitive investment problem. First, it establishes equivalence between the duality-based approach and the Kuroda--Nagai change-of-measure method. Second, the section shows that the change of probability measure induced by duality decomposes into the Kuroda--Nagai transformation followed by an entropy-induced change. Third, this analysis not only recovers the standard optimal investment rule but also clarifies the underlying measure-theoretic and PDE mechanisms. Thus, Section 3 both validates the direct method and highlights the advantages of reframing risk-sensitive control problems in a free energy--entropy duality framework.

Section 4 provides a simple implementation on U.S. equity data. Its purpose is illustrative rather than empirical. Even so, the exercise strengthens the paper in three respects. First, it shows that the duality-based strategy can be computed and simulated in a realistic multi-factor setting. Second, it numerically confirms that the direct duality route and the Kuroda--Nagai route yield identical outcomes, as the theory predicts. Third, by comparing the performance of the resulting allocation with that of the benchmark and of the Kelly portfolio, the model's economic content becomes more concrete and highlights the trade-off between growth and risk control. The implementation is therefore not intended as a forecasting exercise, but as a proof of concept for the mathematical framework developed in the paper.

\section{Solving the Risk-Sensitive Benchmarked Investment Management Problem via the Free Energy--Entropy Duality}\label{sec:RSBAM}

This paper connects stochastic control theory and stochastic differential games, with a view toward a formulation that can later be embedded in a reinforcement-learning framework based on policy-gradient methods. To keep the terminology consistent with that broader objective, we distinguish between \emph{strategy}, \emph{control} (or \emph{action}), and \emph{policy}. A \emph{strategy} is a stochastic process $H=(h_s)_{s\in[0,T]}$ defined on the time interval $[0,T]$. The value $h_s\in\mathbb{R}^m$ denotes the \emph{control}, or \emph{action}, applied at time $s$. Informally, a \emph{policy} specifies how the control is selected as a function of the current time and state. Formally, a policy is a measurable map
\[
h:[0,T]\times\mathbb{R}^n\to\mathbb{R}^m,
\]
and the associated control process is called \emph{Markov} when it takes the form $h_s=h(s,X_s)$--see Remark \ref{rk:Markovcontrol}(i). This terminology is important for the rest of the paper: the control problem will later be recast as a game, the optimal policy will be shown to be Markov, and this same object is the natural one to parameterize in a policy-gradient framework.

Let $\left(\Omega,\mathcal{F},(\mathcal{F}_s)_{0\le s\le T},\mathbb{P}\right)$ be a filtered complete probability space supporting a $d$-dimensional $(\mathcal{F}_s)$-Wiener process $W=(W_s)_{0\le s\le T}$, with components $W_s^i$, $i=1,\ldots,d$. We take the Brownian motion to have dimension $d=n+m+1$, where $n\geq 1$ is the number of factors, $m\geq1$ is the number of financial assets, and the additional dimension captures the component of benchmark risk that is not spanned by the traded assets\footnote{If the benchmark is perfectly spanned by the traded assets, one may instead take $d=n+m$.}.

\subsection{Model for the Financial Market}\label{sec:RSBAM:model}
Using the money market account as num\'eraire, we model the discounted prices of the $m$ risky assets by
\begin{align}\label{eq:dS}
\frac{dS_s^i}{S_s^i}
=
(a_s+A_sX_s)^i\,ds+\sum_{j=1}^d \sigma_s^{ij}\,dW_s^j,
\qquad i=1,\ldots,m,
\end{align}
with initial values $S_0^i$, $i=1,\ldots,m$. The deterministic coefficient functions
$a:[0,T]\to\mathbb{R}^m,\quad
A:[0,T]\to\mathbb{R}^{m \times n},\quad
\Sigma:[0,T]\to\mathbb{R}^{m \times d},$
where $\Sigma_s=(\sigma_s^{ij})_{1\le i\le m,\;1\le j\le d}$, are assumed to be of class $C^1$.

\begin{remark}[Notation]
Throughout the paper, we write $a_s$, $A_s$, $\Sigma_s$, and similarly for other time-dependent coefficients, in place of $a(s)$, $A(s)$, $\Sigma(s)$, and so forth.
\end{remark}

We impose the following nondegeneracy condition on the risky assets.

\begin{assumption}\label{as:sigma:posdef}
For every $s\in[0,T]$, the matrix $\Sigma_s\Sigma_s'$ is positive definite.
\end{assumption}

\begin{remark}
Under Assumption \ref{as:sigma:posdef}, the instantaneous covariance matrix of risky-asset returns is nonsingular, so the traded assets generate distinct directions of risk exposure.
\end{remark}

The discounted asset drifts depend affinely on an $n$-dimensional factor process $X=(X_s)_{s\in[0,T]}$, modeled as a Gaussian diffusion with dynamics
\begin{align}\label{eq:state}
dX_s=(b_s+B_sX_s)\,ds+\Lambda_s\,dW_s,
\end{align}
where
$b:[0,T]\to\mathbb{R}^n,\quad
B:[0,T]\to\mathbb{R}^{n \times n},\quad
\Lambda:[0,T]\to\mathbb{R}^{n \times d}$
are suitable $C^1$ functions. We refer to $X$ as the \emph{state process}. The initial condition is $X_0=x$, where $x$ may be either deterministic or random. In the latter case, we assume $X_0\sim N(\mu,P_0)$, for known $\mu\in\mathbb{R}^n$ and covariance matrix $P_0\in\mathbb{R}^{n \times n}$, and that $X_0$ is independent of the Wiener process $(W_s)_{s\in[0,T]}$.

The investor evaluates portfolio performance relative to a benchmark, such as a market index or a passive reference portfolio. We model the discounted benchmark level $L=(L_s)_{s\in[0,T]}$ by
\begin{align}\label{eq:dL}
\frac{dL_s}{L_s}
=
(c_s+C_sX_s)\,ds+\Xi_s' dW_s,
\end{align}
where
$c:[0,T]\to\mathbb{R},\quad
C:[0,T]\to\mathbb{R}^n,\quad
\Xi:[0,T]\to\mathbb{R}^d$
are $C^1$ functions.

\subsection{Risk-Sensitive Control Problem}\label{sec:RSBAM:RSC_problem}

Let $H=(h_s)_{s\in[0,T]}$ be a $\mathcal{F}_s$-adapted, self-financing investment strategy that represents the vector of portfolio weights in the risky assets. Short selling and leverage are allowed, so $H$ takes its values in $\mathbb{R}^m$. The corresponding discounted wealth process $V=(V_s)_{s\in[0,T]}$ then satisfies
\begin{align}\label{eq:V}
\frac{dV_s}{V_s}
=
h_s'(a_s+A_sX_s)\,ds+h_s'\Sigma_s\,dW_s.
\end{align}

To measure cumulative performance relative to the benchmark, we introduce the log price relative
\[
R_s:=\ln\frac{V_s}{L_s},\qquad s\in[0,T].
\]
Applying It\^o's formula to \eqref{eq:V} and \eqref{eq:dL} yields
\begin{align}\label{eq:excess_return}
dR_s
&=
\ell(s,X_s,h_s)\,ds
+\big(h_s'\Sigma_s-\Xi_s'\big)\,dW_s,
\end{align}
where
\[
\ell(s,x,h)
:=
-\frac{1}{2}h'\Sigma_s\Sigma_s'h
+h'a_s
+\frac{1}{2}\Xi_s'\Xi_s
-c_s
+\big(h'A_s-C_s\big)x.
\]
The cumulative log excess return of the portfolio relative to the benchmark over $[0,T]$ is therefore $R_T-R_0$.

The investor seeks to maximize the risk-sensitive benchmarked performance criterion
\begin{align}\label{eq:J} 
J(H,\theta)
&:= -\frac{1}{\theta} \ln \mathbf{E} \left[ e^{-\theta (R_T-R_0)} \right]
                        \nonumber\\
&= -\frac{1}{\theta} \ln \mathbf{E} \left[ \exp\left\{ -\theta \int_0^T \ell(s,X_s,h_s)\, ds 
  - \theta \int_0^T \left(h_s'\Sigma_s - \Xi_s' \right) dW_s
  \right\}\right],
\end{align}
where $\theta\in(-1,0)\cup(0,\infty)$ is the risk-sensitivity parameter and $T<\infty$ is the investment horizon. Equivalently, one may minimize the exponentially transformed criterion
\begin{align}\label{eq:I}
I(H,\theta)
&:=
e^{-\theta J(H,\theta)}
=
\mathbf{E}\left[e^{-\theta(R_T-R_0)}\right].
\end{align}
In most asset-management applications, the case $\theta>0$ is the primary one of interest. In the benchmarked asset-management interpretation, $\theta>0$ corresponds to more conservative behavior than the Kelly/log-utility benchmark. By contrast, $\theta\in(-1,0)$ leads to overbetting strategies that leverage the Kelly portfolio and therefore trade expected return for additional volatility \citep{davisRisksensitiveBenchmarkedAsset2021}.

We restrict attention to the following admissible class, which is sufficient for the well-posedness of the wealth equation and for the duality arguments developed later.

\begin{definition}\label{def:classAT:fullobs}
A strategy $H=(h_s)_{s\in[0,T]}$ belongs to the class $\mathcal{A}^H$ if:
\begin{enumerate}[(i)]
\item $h_s$ is a Markov control;
\item $\mathbb{P}\!\left(\int_0^T |h_s|^2\,ds<\infty\right)=1$.
\end{enumerate}
\end{definition}

\begin{remark}\label{rk:Markovcontrol}\,
\begin{enumerate}[(i)]
\item A Markov control is a process of the form $h_s=f(s,X_s)$ for some measurable map $f:[0,T]\times\mathbb{R}^n\to\mathbb{R}^m$. We refer the reader to p. 159 in \citet{flso06} for a formal definition. Markov controls are a natural admissibility class for the present problem. Moreover, Theorem \ref{theo:verification} below shows that the optimal strategies derived are indeed generated by affine Markov controls.
\item Condition (ii) imposes square integrability on the control process. This guarantees that the controlled wealth dynamics are well defined and supports the duality arguments developed later in the paper.
\end{enumerate}
\end{remark}

In this setup, the affine factor structure gives tractability while connecting to the empirical asset-pricing and factor-investing literatures. The benchmarked formulation focuses on relative performance and accommodates strategies ranging from active management to passive benchmark replication. Finally, the exponential performance criterion captures risk sensitivity in a form that is compatible with the free energy–entropy duality approach.

\subsection{Free Energy--Entropy Duality Relation}\label{sec:RSBAM:energy-entropy}
We begin by recalling the free energy--entropy duality from \citet{daipraConnectionsStochasticControl1996}, and then show how it applies to our risk-sensitive investment management problem.

\subsubsection{General Definitions}\label{sec:energy-entropy:defs}

The \emph{free energy} of a random variable $\psi$ with respect to the probability measure $\mathbb{P}$, whenever $e^\psi$ is integrable, is defined by
\begin{align}
\mathcal{E}^\mathbb{P}\left(\psi\right)
:=
\ln \left(\int e^\psi\, d\mathbb{P}\right).
\end{align}

Let $\mathbb{Q}$ be another probability measure. The \emph{relative entropy} of $\mathbb{Q}$ with respect to $\mathbb{P}$ is
\begin{align}
D_{\mathrm{KL}}\left(\mathbb{Q} \|\mathbb{P}\right)
:=
\int \ln \left( \frac{d\mathbb{Q}}{d\mathbb{P}} \right)\,d\mathbb{Q}.
\end{align}

The \emph{free energy--entropy duality} relation, as given in \citet[Proposition 2.3(ii)]{daipraConnectionsStochasticControl1996}, states that
\begin{align}
\mathcal{E}^\mathbb{P}\left(\psi\right)
= \sup_{\mathbb{Q} \ll \mathbb{P}} \left\{
\int \psi\, d\mathbb{Q} -  D_{\mathrm{KL}}\left(\mathbb{Q}\|\mathbb{P}\right)
\right\}.
\end{align}

\subsubsection{Applying the Free Energy--Entropy Duality to the Risk-Sensitive Benchmarked Investment Problem}\label{sec:energy-entropy:DS}

Letting $\psi = -\theta (R_T-R_0)$, the duality relation implies that, for every $H \in \mathcal{A}^H$,
\begin{align}\label{eq:RS:duality:I:interim}
    \ln I(H,\theta) 
    = \ln \mathbf{E} \left[ e^{-\theta (R_T-R_0)} \right]
    = \sup_{\mathbb{P}^\Gamma \ll \mathbb{P}} \left\{\mathbf{E}^{\mathbb{P}^\Gamma} \left[ -\theta (R_T-R_0) \right] 
    - D_{\mathrm{KL}}\left(\mathbb{P}^\Gamma \| \mathbb{P} \right) \right\},
\end{align}
where the supremum is taken over probability measures absolutely continuous with respect to $\mathbb{P}$; see \citet[Section 3.1.2]{daipraConnectionsStochasticControl1996}. Here, $\mathbf{E}^{\mathbb{P}^\Gamma}[\cdot]$ denotes expectation under $\mathbb{P}^\Gamma$.

We now make the variational representation in \eqref{eq:RS:duality:I:interim} explicit by introducing a process $\Gamma$ that generates the measure change. By Theorem 7.11 in \citet{LiShI04}, there exists an $(\mathcal{F}_s)$-progressively measurable process $\Gamma=(\gamma_s)_{s\in[0,T]}$ in $\mathbb{R}^d$ such that $\int_0^T \|\gamma_s\|^2 ds<\infty$ almost surely and
\begin{align}\label{eq:Wgamma}
W^\Gamma_s := W_s - \int_0^s \gamma_u\,du
\end{align}
is a $\mathbb{P}^\Gamma$-Wiener process. Define the density process
\begin{align}\label{eq:RNderivative:gamma}
    \chi^\Gamma_s := \exp \left\{ -\frac{1}{2} \int_0^s  \| \gamma_u \|^2 du + \int_0^s \gamma_u' dW_u \right\},    \quad s \in [0,T].
\end{align} 
By Girsanov's theorem, $\frac{d\mathbb{P}^\Gamma}{d\mathbb{P}}\Big{|}_{\mathcal{F}_T} = \chi^\Gamma_T$, 
where $\frac{d\mathbb{P}^\Gamma}{d\mathbb{P}}\Big{|}_{\mathcal{F}_T}$ is the restriction of the Radon-Nikodym derivative $\frac{d\mathbb{P}^\Gamma}{d\mathbb{P}}$ to the filtration $\mathcal{F}_T$.

Under $\mathbb{P}^\Gamma$, $W^\Gamma$ is a Brownian motion, and the drift terms are shifted accordingly. Under $\mathbb{P}^\Gamma$, the state process satisfies
\begin{align}\label{eq:state:Pgamma:FO}
dX_s
=
\left(b_s+B_sX_s+\Lambda_s\gamma_s\right)\,ds
+
\Lambda_s\,dW_s^\Gamma.
\end{align}
The corresponding dynamics of the log price relative process become
\begin{align}\label{eq:excess_return:Pgamma:FO}
dR_s
&=
\Big[
-\frac12 h_s'\Sigma_s\Sigma_s'h_s
+h_s'a_s
+\frac12 \Xi_s'\Xi_s
-c_s
+\left(h_s'\Sigma_s-\Xi_s'\right)\gamma_s
+\left(h_s'A_s-C_s\right)X_s
\Big]\,ds
\nonumber\\
&\qquad
+
\left(h_s'\Sigma_s-\Xi_s'\right)\,dW_s^\Gamma.
\end{align}

Moreover, standard entropy calculations yield
\begin{align}\label{eq:DKL:Pgamma:P}
D_{\mathrm{KL}}(\mathbb{P}^\Gamma\|\mathbb{P})
=
\mathbf{E}^{\mathbb{P}^\Gamma}\left[\ln \chi_T^\Gamma\right]
=
\frac12\mathbf{E}^{\mathbb{P}^\Gamma}\left[\int_0^T \|\gamma_s\|^2\,ds\right].
\end{align}

Therefore, we can express the \emph{free energy--entropy duality} as
\begin{align}\label{eq:RS:duality:I}
    \ln I(H,\theta) = \sup_{ \Gamma} \mathbf{E}^{\mathbb{P}^\Gamma} \left[ -\theta (R_T-R_0) 
    - \frac{1}{2} \int_0^T \|\gamma_s\|^2 ds  \right],
\end{align}
for all choices of $H \in\mathcal{A}^H$, and where the supremum is taken over progressively measurable processes $\Gamma = \left(\gamma_s\right)_{s \in [0,T]}$.

To prepare the dynamic programming analysis, we restrict attention to a Markovian admissible class for the maximizing player.

\begin{definition}\label{def:classAgammaT:fullobs}

A \emph{strategy} $\Gamma = \left(\gamma_s\right)_{s \in [0,T]} \in \mathbb{R}^{d}$  is in class $\mathcal{A}^\Gamma$ if the \emph{control} process $\gamma_s$ satisfies the following conditions:

\begin{enumerate}[(i)]

    \item  $\gamma_s$ is a Markov control;

    \item $\mathbb{P} \left(\int_{0}^{T} \left| \gamma_s \right|^2 ds < +\infty \right) = 1$;

    \item $(\chi^\Gamma_s)_{s \in [0,T]}$, defined at \eqref{eq:RNderivative:gamma} above, is an exponential martingale.

\end{enumerate}

\end{definition}
The Markovian structure of the problem suggests that an optimizer $\Gamma^*$ should lie in $\mathcal{A}^\Gamma$. We therefore impose the following restriction, which is later validated by the verification theorem.
\begin{assumption}\label{as:Gamma:Markov}
    \begin{align}\label{eq:RS:duality:I:restriction}
    \sup_{ \Gamma } \mathbf{E}^{\mathbb{P}^\Gamma} \left[ -\theta (R_T-R_0) 
    - \frac{1}{2} \int_0^T \|\gamma_s\|^2 ds  \right]
    =\sup_{ \Gamma \in \mathcal{A}^\Gamma} \mathbf{E}^{\mathbb{P}^\Gamma} \left[ -\theta (R_T-R_0) 
    - \frac{1}{2} \int_0^T \|\gamma_s\|^2 ds  \right].
\end{align}
\end{assumption}

Expanding \eqref{eq:RS:duality:I} using the dynamics of $R$ under $\mathbb{P}^\Gamma$, we obtain
\begin{align}\label{eq:RS:duality:I:2}
    & \ln I(H,\theta) 
                           \nonumber\\
    =& \sup_{ \Gamma} \mathbf{E}^{\mathbb{P}^\Gamma} \Bigg[ 
    -\theta \int_0^T \left\{
        - \frac{1}{2} h_s'\Sigma_s\Sigma_s'h_s 
        + h_s' a_s 
        + \frac{1}{2}\Xi_s'\Xi_s - c_s  
        + \left(h_s'\Sigma_s - \Xi_s' \right)\gamma_s
        + \left(h_s' A_s - C_s \right)X_s \right\}ds
                            \nonumber\\
    & - \theta \int_0^T \left(h_s'\Sigma_s - \Xi_s' \right) dW^\Gamma_s
    - \frac{1}{2} \int_0^T \|\gamma_s\|^2 ds  
    \Bigg]
                                \nonumber\\
    \underset{\text{Assumption }\ref{as:Gamma:Markov}}{=}& \sup_{ \Gamma \in \mathcal{A}^\Gamma} \mathbf{E}^{\mathbb{P}^\Gamma} \Bigg[ 
    \theta \int_0^T g\left(s,X_s,h_s,\gamma_s;\theta\right) ds \Bigg],
\end{align}
where the running payoff $g$ is defined by
\begin{align}\label{eq:g}
    g(s,x,h,\gamma;\theta) 
    :=  \frac{1}{2} h'\Sigma_s\Sigma_s'h
        - h' a_s 
        - \frac{1}{2}\Xi_s'\Xi_s 
        + c_s  
        - \left(h'\Sigma_s - \Xi_s' \right)\gamma
        - \left(h' A_s - C_s \right)x 
        - \frac{1}{2\theta} \|\gamma\|^2.
\end{align}
The last equality follows because the stochastic integral, $\int_{0}^{T} \left(h_s' \Sigma_s - \Xi_s' \right) dW^{\Gamma}_s$, is a $\mathbb{P}^\Gamma$-martingale for $H\in\mathcal{A}^H$ and $\Gamma\in\mathcal{A}^\Gamma$.

In what follows, we focus on the practically most relevant case $\theta>0$. The case $\theta\in(-1,0)$ can be handled similarly\footnote{For the criterion $J$, instead of maximizing a concave function when $\theta>0$, we minimize a convex function when $\theta<0$.}, but it leads to overbetting relative to the Kelly benchmark and is therefore of limited practical interest. Since the logarithm is strictly increasing, minimizing $I(H,\theta)$ is equivalent to minimizing $\ln I(H,\theta)$. We make this equivalence explicit in the next display.
\begin{align}\label{eq:EEDuality:inf}
    & \inf_{H \in \mathcal{A}^H}\ln I(H,\theta) 
    = \inf_{H \in \mathcal{A}^H} \sup_{ \Gamma \in \mathcal{A}^\Gamma} \mathbf{E}^{\mathbb{P}^\Gamma} \left[ \theta \int_{0}^{T} g(s,X_s,h_s,\gamma_s;\theta) ds
    \right]
                            \nonumber\\
    \Leftrightarrow&
    \ln \inf_{H \in \mathcal{A}^H} I(H,\theta) 
    = \inf_{H \in \mathcal{A}^H} \sup_{ \Gamma \in \mathcal{A}^\Gamma} \mathbf{E}^{\mathbb{P}^\Gamma} \left[ \theta \int_{0}^{T} g(s,X_s,h_s,\gamma_s;\theta) ds  \right]
                            \nonumber\\
    \Leftrightarrow&
    \inf_{H \in \mathcal{A}^H} I(H,\theta) 
    = \exp\left\{ \inf_{H \in \mathcal{A}^H} \sup_{ \Gamma \in \mathcal{A}^\Gamma} \mathbf{E}^{\mathbb{P}^\Gamma} \left[ \theta \int_{0}^{T} g(s,X_s,h_s,\gamma_s;\theta) ds
    \right] 
    \right\}
,
\end{align}
where the first equivalence follows from Lemma \ref{lem:infexp:app} in Appendix \ref{app:Lemma:Meneghini}.

The right-hand side of \eqref{eq:EEDuality:inf} admits a natural interpretation as a two-player stochastic differential game, which we now define.

\begin{definition}[The implied stochastic differential game]\label{game}
For $H=(h_s)_{s\in[0,T]}\in\mathcal{A}^{H}$ and $\Gamma=(\gamma_s)_{s\in[0,T]}\in\mathcal{A}^\Gamma$, the implied stochastic differential game is defined by
\begin{align}\label{eq:def:SDG}
\inf_{H \in \mathcal{A}^H} \sup_{ \Gamma \in \mathcal{A}^\Gamma} \mathbf{E}^{\mathbb{P}^\Gamma} \left[ \theta \int_{0}^{T} g(s,X_s,h_s,\gamma_s;\theta) ds
\right],
\end{align}
where $H$ is the minimizing player's strategy, $\Gamma$ is the maximizing player's strategy, and the state dynamics are given by \eqref{eq:state:Pgamma:FO}.
\end{definition}

\subsubsection{Solving the Stochastic Differential Game}\label{S.2.3.3}

This subsection studies the stochastic differential game introduced in Definition \ref{game}. The game is linear-quadratic-Gaussian. We show that the optimal strategies are generated by affine policies of the state variable $X$, and hence are Markov controls. Under Markov controls, the controlled state process remains Markov. Therefore, at any intermediate time $t\in[0,T)$, conditional expectations of functionals of $(X_s,h_s,\gamma_s)_{s\in[t,T]}$, given $X_t=x$, can be expressed as functions of $(t,x)$. We then show that the game's value function, defined in \eqref{eq:valuefunction:u}, is quadratic in the state variable 
$x$, with time-varying coefficients.

As in the previous subsection, we assume that $\theta>0$ and consider a generic starting time $t\in[0, T)$. For convenience, we normalize the benchmark level at time $t$ by setting $L_t=V_t$, so that the corresponding log price relative satisfies $R_t=0$. The corresponding continuation criteria are then
\begin{align}
    J(t,x;H,\theta) :=& -\frac{1}{\theta} \ln \mathbf{E}_{t,x} \left[ e^{-\theta R_T} \right]    \label{eq:J:general}\\
    I(t,x;H,\theta) 
    :=& \; e^{-\theta J(t,x;H,\theta)} 
    = \mathbf{E}_{t,x} \left[ e^{-\theta R_T} \right].
                    \label{eq:I:general}
\end{align}
The continuation version of the last relation in \eqref{eq:EEDuality:inf} is
\begin{align}\label{eq:EEDuality:inf:SDG} 
\inf_{H \in \mathcal{A}^H} I(t,x;H,\theta) 
    = \exp\left\{ \inf_{H \in \mathcal{A}^H} \sup_{ \Gamma \in \mathcal{A}^\Gamma} \mathbf{E}_{t,x}^{\mathbb{P}^\Gamma} \left[ \theta \int_{t}^{T} g(s,X_s,h_s,\gamma_s;\theta) ds
    \right]\right\},
\end{align}
which yields the continuation version of the stochastic differential game:
\begin{align}\label{eq:EEDuality:SDG:general}
    \inf_{H \in \mathcal{A}^H} \sup_{ \Gamma \in \mathcal{A}^\Gamma} \mathbf{E}_{t,x}^{\mathbb{P}^\Gamma} \left[ \theta \int_{t}^{T} g(s,X_s,h_s,\gamma_s;\theta) ds
    \right].
\end{align}
We define the associated value function $u(t,x)$ by
\begin{align}\label{eq:valuefunction:u}
    u(t,x) := \inf_{H \in \mathcal{A}^H} \sup_{ \Gamma \in \mathcal{A}^\Gamma} \mathbf{E}_{t,x}^{\mathbb{P}^\Gamma} \left[ \theta \int_{t}^{T} g(s,X_s,h_s,\gamma_s;\theta) ds \right].
\end{align}

\begin{lemma}\label{lem:valuefunctions}
The value function satisfies
\begin{align}\label{eq:u_and_utilde}
    u(t,x) 
    = \ln \inf_{H \in \mathcal{A}^H}I(t,x;H,\theta)
    = - \theta \sup_{H \in \mathcal{A}^H} J(t,x;H,\theta).
\end{align}
\end{lemma}

\begin{proof}
The first equality follows directly from the definition of $u$ in \eqref{eq:valuefunction:u} together with \eqref{eq:EEDuality:inf:SDG}. For the second equality, use \eqref{eq:I:general} and the identity
\[
\inf_{H \in \mathcal{A}^H} \big(-\theta J(t,x;H,\theta)\big)
=
-\theta \sup_{H \in \mathcal{A}^H} J(t,x;H,\theta),
\qquad \theta>0.
\]
\end{proof}

\begin{remark}\label{rq:EED:control}
The free energy--entropy duality rewrites the original risk-sensitive investment criterion under $\mathbb{P}$ as the objective of an LQG game under $\mathbb{P}^\Gamma$. The running payoff $g$ depends on both the investment strategy $H$ and the change-of-measure process $\Gamma$. In particular, it contains the interaction term $h_s'\Sigma_s\gamma_s$. By contrast, the state dynamics in \eqref{eq:state:Pgamma:FO} are controlled only by the change-of-measure process $\Gamma$.
\end{remark}

The function $g$ is affine in the state variable $x$ and quadratic in the control variables $h$ and $\gamma$. Hence, the stochastic differential game associated with $u(t,x)$ is a linear-quadratic-Gaussian (LQG) game. The associated Bellman--Isaacs partial differential equation is
\begin{align}\label{eq:BellmanIsaacs:u}
    \frac{\partial u(s,x)}{\partial s}
    + \mathcal{H}^+\left(s,x,Du(s,x), D^2 u(s,x)\right)
    = 0,
\end{align}
with terminal condition $u(T,x)=0$. Here, $Du(s,x)=\left(\frac{\partial u}{\partial x_1}(s,x),\ldots,\frac{\partial u}{\partial x_n}(s,x)\right)'$
and
$D^2u(s,x)=\left[\frac{\partial^2 u}{\partial x_i\partial x_j}(s,x)\right]_{i,j=1}^n.$

The Hamiltonian is given by
\begin{align}\label{eq:Hamiltonian:H+}
   \mathcal{H}^+(s,x,p,M)
   := \inf_{h \in \mathbb{R}^m} \sup_{\gamma \in \mathbb{R}^d} \left\{ \left[b_s + B_s x + \Lambda_s \gamma \right]' p 
    + \frac{1}{2} \tr \left(\Lambda_s\Lambda_s' M \right)
    +\theta g(s,x,h,\gamma;\theta) \right\}.
\end{align}

For completeness, we also define the Hamiltonian 
\begin{align}\label{eq:Hamiltonian:H-}
   \mathcal{H}^-(s,x, p, M) := \sup_{\gamma \in \mathbb{R}^d}\inf_{h \in \mathbb{R}^m} \left\{ \left[b_s + B_s x + \Lambda_s \gamma \right]' p 
    + \frac{1}{2} \tr \left(\Lambda_s\Lambda_s' M \right)
    +\theta g(s,x,h,\gamma;\theta) \right\}. 
\end{align}

For notational convenience, let
\begin{align}
    \mathcal{P}_s^{+}(\theta)
    &:= I_d + \theta \Sigma_s'(\Sigma_s\Sigma_s')^{-1} \Sigma_s,
                    \label{def:calPplus}\\
    \mathcal{P}_s^{-}(\theta)
    &:= I_d - \frac{\theta}{\theta+1}\Sigma_s'(\Sigma_s\Sigma_s')^{-1} \Sigma_s.
                    \label{def:calPminus}
\end{align}
If we define $\Pi_s := \Sigma_s'(\Sigma_s\Sigma_s')^{-1}\Sigma_s$, then $\Pi_s^2=\Pi_s$, so $\mathcal{P}_s^{-}(\theta)\mathcal{P}_s^{+}(\theta)=I_d$. Hence,
\begin{align}\label{eq:projinverse}
\left(\mathcal{P}_s^{+}(\theta)\right)^{-1}=\mathcal{P}_s^{-}(\theta).
\end{align}

\begin{proposition}\label{prop:candidate_controls}
The Bellman--Isaacs equation \eqref{eq:BellmanIsaacs:u} yields a unique pair of candidate feedback policies $\hat h$ and $\hat\gamma$, given equivalently by either of the following two representations:
\begin{align}
        \hat{h}(s,x,Du(s,x))
            =&  \frac{1}{\theta+1}\left(\Sigma_s\Sigma_s'\right)^{-1}
        \left( a_s + A_s x
        + \theta \Sigma_s\Xi_s 
        + \Sigma_s\Lambda_s' Du(s,x)
        \right)
                            \label{eq:hhat}\\    
    \hat{\gamma}(s,x,Du(s,x)) 
=& \Lambda_s' Du(s,x) - \theta \left(\Sigma_s'\hat{h}(s,x,Du(s,x)) - \Xi_s \right)
                        \label{eq:gammahat:1}\\
    \phantom{\hat{\gamma}(s,x,Du(s,x))}
=&  \mathcal{P}_s^{-}(\theta)\Lambda_s' Du(s,x) 
        - \frac{\theta}{\theta+1} \Sigma_s'\left(\Sigma_s\Sigma_s'\right)^{-1}
        \left( a_s + A_s x \right) 
        + \theta \mathcal{P}_s^{-}(\theta) \Xi_s,
                        \label{eq:gammahat:2}
    \end{align}   
    or 
     \begin{align}
        \hat{h}(s,x,Du(s,x))
        =&  \left(\Sigma_s\Sigma_s'\right)^{-1}  
        (a_s + A_s x) + \left(\Sigma_s\Sigma_s'\right)^{-1}\Sigma_s \hat{\gamma}(s,x,Du(s,x))
                                                \label{eq:hhat:alt}\\
        \hat{\gamma}(s,x,Du(s,x)) 
    =& \mathcal{P}_s^{-}(\theta)
        \left[ 
            - \theta \Sigma_s' \left(\Sigma_s\Sigma_s'\right)^{-1} (a_s + A_s x)
            + \theta \Xi_s
            + \Lambda_s'Du(s,x)
        \right]
                        \label{eq:gammahat:alt}
    \end{align}  
     Moreover, the two Hamiltonians coincide: $\mathcal{H}^+(s,x, p, M) = \mathcal{H}^-(s,x, p, M)$. Therefore, the Isaacs condition holds, and the stochastic differential game is well defined.
\end{proposition}

Technically, we only need to solve the problem \eqref{eq:BellmanIsaacs:u} with the Hamiltonian $\mathcal{H}^+$ given at \eqref{eq:Hamiltonian:H+}.  Therefore, Proposition \ref{prop:candidate_controls}  gives us a stronger result than required and justifies the interpretation of the regularized problem as a stochastic differential game. 

\begin{proof} 
See Appendix \ref{app:proof:prop:candidate_controls}.
\end{proof}

\begin{remark}\label{rk:candidate:saddle}
The candidate policies satisfy the following local saddle-point condition:
\begin{align}\label{eq:candidate:saddle}
&\left[b_s + B_s x + \Lambda_s \gamma \right]' Du(s,x) 
    + \frac{1}{2} \tr \left(\Lambda_s\Lambda_s' D^2u(s,x) \right)
    +\theta g\left(s,x,\hat{h}(s,x,Du(s,x)),\gamma;\theta\right)
                                        \nonumber\\
\leq&
\left[b_s + B_s x + \Lambda_s \hat{\gamma}(s,x,Du(s,x)) \right]' Du(s,x) 
    + \frac{1}{2} \tr \left(\Lambda_s\Lambda_s' D^2u(s,x) \right)
                                        \nonumber\\
&\quad
    +\theta g\left(s,x,\hat{h}(s,x,Du(s,x)),\hat{\gamma}(s,x,Du(s,x));\theta\right)
                                    \nonumber\\
\leq&
\left[b_s + B_s x + \Lambda_s \hat{\gamma}(s,x,Du(s,x)) \right]' Du(s,x) 
    + \frac{1}{2} \tr \left(\Lambda_s\Lambda_s' D^2u(s,x) \right)
    +\theta g\left(s,x,h,\hat{\gamma}(s,x,Du(s,x));\theta\right),
\end{align}
for all $h\in\mathbb{R}^m$ and $\gamma\in\mathbb{R}^d$. We use this pointwise saddle-point property later in the proof of Verification Theorem \ref{theo:verification} to establish the corresponding global saddle-point condition.
\end{remark}

\begin{remark}[Measure $\mathbb{P}^{\hat\Gamma}$]
Since the candidate control strategies in Proposition \ref{prop:candidate_controls} admit equivalent representations, we use here the first pair \eqref{eq:hhat}--\eqref{eq:gammahat:1}. The corresponding density increment over $[t,T]$ is
\begin{align}\label{eq:RNderivative:gammahat}
\chi^{\hat\Gamma}_{[t,T]}
:=& \exp \bigg\{ -\frac{1}{2} \int_t^T  \bigg\| \Lambda_s' Du(s,X_s) - \theta \left(\Sigma_s'\hat{h}(s,X_s,Du(s,X_s)) - \Xi_s \right) \bigg\|^2 ds
                                                \nonumber\\
     & +\int_t^T \left[\Lambda_s' Du(s,X_s) - \theta \left(\Sigma_s'\hat{h}(s,X_s,Du(s,X_s)) - \Xi_s \right)\right]' dW_s \bigg\}.
\end{align}
Substituting \eqref{eq:gammahat:2} yields the equivalent expression
\begin{align*}
\chi^{\hat\Gamma}_{[t,T]}
=& \exp \Bigg\{ -\frac{1}{2} \int_t^T  \bigg\| 
        - \frac{\theta}{\theta+1}\Sigma_s'\left(\Sigma_s\Sigma_s'\right)^{-1}
        \left( a_s + A_s X_s \right)
        + \theta \mathcal{P}_s^{-}(\theta)\Xi_s 
        + \mathcal{P}_s^{-}(\theta)\Lambda_s' Du(s,X_s)
        \bigg\|^2 ds 
                                                \nonumber\\
     & + \int_t^T \Bigg[
        - \frac{\theta}{\theta+1}\Sigma_s'\left(\Sigma_s\Sigma_s'\right)^{-1}
        \left( a_s + A_s X_s \right)
        + \theta \mathcal{P}_s^{-}(\theta)\Xi_s 
        + \mathcal{P}_s^{-}(\theta)\Lambda_s' Du(s,X_s)
            \Bigg]' dW_s \Bigg\}.
\end{align*}
\end{remark}

\begin{remark}[Special case of the Kelly portfolio]\label{rk:Kelly:1}
Formally, as $\theta \to 0$, one recovers the Kelly or growth-optimal portfolio. In that limit, the density increment over $[t,T]$ simplifies to
\begin{align}\label{eq:RNderivative:gammahat:theta_so_0}
    \chi^{\hat\Gamma}_{[t,T]}
     = \exp \left\{ -\frac{1}{2} \int_t^T  \| \Lambda_s' Du(s,X_s) \|^2 ds + \int_t^T \left[\Lambda_s' Du(s,X_s) \right]' dW_s \right\}.
\end{align}
\end{remark}

\begin{theorem}\label{theo:sol_IsaacPDE}
The Bellman--Isaacs equation \eqref{eq:BellmanIsaacs:u}, associated with the candidate strategy pair $(\hat H,\hat\Gamma)$ defined in Proposition \ref{prop:candidate_controls}, admits a solution of the form 
\begin{align}\label{eq:sol:u}
u(t,x)=-\theta\left(\frac12 x'Q_t x + q_t'x + k_t\right),
\end{align}
where $Q:[0,T]\to\mathbb R^{n \times n}$, $q:[0,T]\to\mathbb R^n$, and $k:[0,T]\to\mathbb R$ are deterministic functions such that:
    \begin{enumerate}[1.]
    \item $Q$ is the unique symmetric positive semidefinite solution to the matrix Riccati equation
    \begin{align}\label{eq:Q:Riccati}
        &\dot{Q}_s
            - \theta  Q_s\Lambda_s \mathcal{P}_s^{-}(\theta) \Lambda_s'Q_s
            + \left(B_s' - \frac{\theta}{(\theta+1)} A_s'(\Sigma_s\Sigma_s')^{-1} \Sigma_s\Lambda_s' 
            \right)Q_s 
                                            \nonumber\\
        &+ Q_s\left(B_s - \frac{\theta}{(\theta+1)} \Lambda_s\Sigma_s'(\Sigma_s\Sigma_s')^{-1}A_s \right)
            + \frac{1}{(\theta+1)} A_s' (\Sigma_s\Sigma_s')^{-1} A_s
        = 0, \quad Q(T)=0;
    \end{align}
    \item $q$ solves the linear ODE
    \begin{align}\label{eq:q:ODE}
        & \dot{q}_s
        + \left(B_s'- \frac{\theta}{\theta+1} A_s'(\Sigma_s\Sigma_s')^{-1}\Sigma_s\Lambda_s' \right)q_s
        - \theta Q_s \Lambda_s\mathcal{P}_s^{-}(\theta)\Lambda_s' q_s
                                        \nonumber\\
        &   + Q_s \left[b_s - \frac{\theta}{\theta+1}  \Lambda_s\Sigma_s'(\Sigma_s\Sigma_s')^{-1}
    \left( a_s + \theta \Sigma_s\Xi_s \right) + \theta \Lambda_s\Xi_s\right]
                                        \nonumber\\
        &   - C_s'
        + \frac{1}{\theta+1} A_s'(\Sigma_s\Sigma_s')^{-1}
    \left( a_s + \theta \Sigma_s\Xi_s \right)
        = 0, \quad q(T)=0.
    \end{align}
    \item $k$ is given by the integral
    \begin{align}\label{eq:k:integral}
        & k_t
        = \int_t^T \Bigg\{ 
        \frac{\theta}{2} q_s'\Lambda_s \mathcal{P}_s^{-}(\theta)\Lambda_s' q_s
        - \left[b_s'+ \theta \Xi_s'\Lambda_s' - \frac{\theta}{\theta+1} \left( a_s + \theta \Sigma_s\Xi_s \right)'(\Sigma_s\Sigma_s')^{-1}\Sigma_s\Lambda_s' \right]q_s 
        - \frac{1}{2} \tr \left(\Lambda_s\Lambda_s' Q_s \right)
                                        \nonumber\\
        &   - \frac{1}{2(\theta+1)} \left(a_s + \theta \Sigma_s\Xi_s\right)'(\Sigma_s\Sigma_s')^{-1} \left(a_s + \theta \Sigma_s\Xi_s\right)
        + c_s
        + \frac{\theta-1}{2}\Xi_s'\Xi_s \Bigg\} ds.                               
    \end{align}    
    \end{enumerate}    
\end{theorem}

\begin{proof} 
See Appendix \ref{app:proof:theo:sol_IsaacPDE}.
\end{proof}

\begin{remark}\label{rk:gammahat:nonlinearterm} 
The term $\Lambda_s'Du(s,x)$ in the definition of $\hat{\gamma}$ at \eqref{eq:gammahat:1} generates, after optimization over $\gamma$, a quadratic gradient term in the Bellman--Isaacs PDE. In the present benchmarked risk-sensitive setting, this term appears through the matrix
$\Lambda_s \mathcal{P}_s^{-}(\theta) \Lambda_s'$. Consequently, the equations \eqref{eq:Q:Riccati}--\eqref{eq:k:integral} contain terms that are characteristic of risk-sensitive control problems and absent from standard LQG problems, namely, 
$-\theta Q_s \Lambda_s\mathcal{P}_s^-(\theta)\Lambda_s' Q_s$ in \eqref{eq:Q:Riccati}, $-\theta Q_s \Lambda_s\mathcal{P}_s^-(\theta)\Lambda_s' q_s$ in \eqref{eq:q:ODE}, and $\frac{\theta}{2} q_s'\Lambda_s \mathcal{P}_s^{-}(\theta) \Lambda_s' q_s$ in \eqref{eq:k:integral}.
\end{remark}

\begin{theorem}[Verification]\label{theo:verification}
\;
Let $u$ be the value function defined at \eqref{eq:valuefunction:u} and having the quadratic expression in \eqref{eq:sol:u}. Let $H^* = \left(h^*(s,X_s) \right)_{s \in [t,T]} := \left(\hat{h}(s,X_s,Du(s,X_s)) \right)_{s \in [t,T]}$ and $\Gamma^* = \left(\gamma^*(s,X_s) \right)_{s \in [t,T]}  := \left(\hat{\gamma}(s,X_s,Du(s,X_s)) \right)_{s \in [t,T]}$, where $\hat{h}(\cdot)$ and $\hat{\gamma}(\cdot)$ are the candidate policies given in Proposition \ref{prop:candidate_controls}, and let 
\[
h_s^* := h^*(s,X_s), \qquad \gamma_s^* := \gamma^*(s,X_s).
\] 
Then,
\begin{enumerate}[1.]
    \item The pair of strategies $(H^*,\Gamma^*)$ is admissible in the sense that they belong to the classes $\mathcal{A}^H$ and $\mathcal{A}^{\Gamma}$ respectively. Moreover, their defining policies $h^*(\cdot)$ and $\gamma^*(\cdot)$ are Borel-measurable maps.

       \item The value function $u$ solves the stochastic differential game with the objective \eqref{eq:EEDuality:SDG:general}, in the sense that the saddle point condition
        \begin{align}\label{eq:verification:saddle}
            \mathbf{E}_{t,x}^{\mathbb{P}^\Gamma} \left[ \theta \int_{t}^{T} g(s,X_s,h^*_s,\gamma_s;\theta) ds \right]
            \leq 
            u(t,x) 
            \leq 
            \mathbf{E}_{t,x}^{\mathbb{P}^{\Gamma^*}} \left[ \theta \int_{t}^{T} g(s,X_s,h_s,\gamma^*_s;\theta) ds \right]
        \end{align}
        holds for any admissible strategies $H = \left(h_s\right)_{s\in[t,T]} \in \mathcal{A}^H$ and $\Gamma = \left(\gamma_s\right)_{s \in [t,T]} \in \mathcal{A}^\Gamma$, and that
        \begin{align}
            u(t,x) =\mathbf{E}_{t,x}^{\mathbb{P}^{\Gamma^*}} \left[ \theta \int_{t}^{T} g(s,X_s,h^*_s,\gamma^*_s;\theta) ds \right].
        \end{align} 
    
    \item The strategies $H^*$ and $\Gamma^* $ are optimal.

\end{enumerate}
    
\end{theorem}

\begin{proof}
See Appendix \ref{app:proof:theo:verif}.
\end{proof}

\begin{remark}
    The pair of strategies $(H^*, \Gamma^*)$ forms a \emph{saddle point} of the stochastic differential game, as evidenced by  Remark \ref{rk:candidate:saddle} and equation \eqref{eq:verification:saddle}. Equivalently, it is a \emph{Nash equilibrium in pure strategies} for the associated stochastic game.
\end{remark}

We conclude this section with a well-known decomposition\footnote{We refer the reader to \citep{dall_RSBench} and Chapter 3 in \citep{DavisLleoBook2014}, among others.} of the optimal investment strategy as a fractional Kelly strategy. The strategy combines three constituent portfolios: the Kelly or log-optimal portfolio, a benchmark-tracking portfolio, and an intertemporal hedging portfolio.

\begin{proposition}[Fractional Kelly strategy]\label{prop:PFKS}
Let $f:=\frac{1}{\theta+1}$. Then, for $s\in[t,T]$, the optimal investment strategy admits the decomposition
\begin{align}\label{eq:PFKS:decomp}
h^*(s,X_s)
=
f\,h^K(s,X_s)
+
(1-f)\,h_s^\mathrm{Bench}
-
(1-f)\,h^\mathrm{IHP}(s,X_s),
\end{align}
where
\begin{align}
h^K(s,X_s)
&=
\left(\Sigma_s\Sigma_s'\right)^{-1}\left(a_s+A_sX_s\right),
\label{eq:Kelly}\\
h_s^\mathrm{Bench}
&=
\left(\Sigma_s\Sigma_s'\right)^{-1}\Sigma_s\Xi_s,
\label{eq:BenchTracking}\\
h^\mathrm{IHP}(s,X_s)
&=
\left(\Sigma_s\Sigma_s'\right)^{-1}\Sigma_s\Lambda_s'\left(q_s+Q_sX_s\right).
\label{eq:IHP}
\end{align}
\end{proposition}

\begin{proof}
Using \eqref{eq:hhat} together with $Du(s,X_s)=-\theta\left(Q_sX_s+q_s\right)$,    
we obtain
\begin{align*}
h^*(s,X_s)
=
\frac{1}{\theta+1}\left(\Sigma_s\Sigma_s'\right)^{-1}(a_s+A_sX_s)
+
\frac{\theta}{\theta+1}\left(\Sigma_s\Sigma_s'\right)^{-1}\Sigma_s\Xi_s
-
\frac{\theta}{\theta+1}\left(\Sigma_s\Sigma_s'\right)^{-1}\Sigma_s\Lambda_s'(q_s+Q_sX_s),    
\end{align*}
which is exactly \eqref{eq:PFKS:decomp} for the definitions \eqref{eq:Kelly}--\eqref{eq:IHP}.

\end{proof}

The free energy--entropy duality also yields a complementary interpretation of the optimal allocation. In addition to the fractional Kelly decomposition above, the optimal strategy can be written as the Kelly portfolio corrected by a term induced by the optimal change-of-measure control $\gamma^*$, as shown in the next corollary.

\begin{corollary}[Optimal Allocation as a Regularized Kelly Strategy]\label{coro:regularizedKelly}
The optimal investment strategy satisfies
\begin{align}\label{eq:regularized_Kelly_strategy}
h^*(s,X_s)
=
h^K(s,X_s)
+
\left(\Sigma_s\Sigma_s'\right)^{-1}\Sigma_s\,\gamma^*(s,X_s).
\end{align}
\end{corollary}

\begin{proof}
The identity follows immediately from \eqref{eq:hhat:alt} and the definition of $h^K$ in \eqref{eq:Kelly}.

\end{proof}

Equation \eqref{eq:regularized_Kelly_strategy} mirrors the structure of the free energy--entropy duality. At the criterion level, the duality rewrites the original risk-sensitive problem as an LQG problem under a transformed probability measure, with an additional quadratic penalty on the change-of-measure control $\gamma$. At the control level, the optimal allocation is the Kelly (or log-optimal) strategy corrected by the term $(\Sigma_s\Sigma_s')^{-1}\Sigma_s\,\gamma^*(s,X_s)$ induced by the optimal antagonistic control $\gamma^*$. This correction also admits a natural geometric interpretation: it is the image of the optimal antagonistic control $\gamma^*(s,X_s)\in\mathbb{R}^d$ under the linear map $(\Sigma_s\Sigma_s')^{-1}\Sigma_s$.

\section{Understanding the Free Energy–Entropy Duality Through the Lens of the Kuroda–Nagai Change of Probability Measure}\label{sec:Kuroda--Nagai}

This section analyzes the role of the change of measure in the application of the free energy–entropy duality to the benchmarked risk-sensitive investment management problem. To this end, we compare the direct one-step approach developed in the previous section with a two-step route in which one first performs a change of measure \emph{\`a la} \citet{kuna02} and then applies the duality under the transformed probability measure. Because the Kuroda–Nagai measure change underlies many of the standard solution techniques in the risk-sensitive investment management literature \citep[see, e.g.,][and the references therein]{DavisLleoBook2014,davisRisksensitiveBenchmarkedAsset2021,dall_JDBenchAltData2024}, it provides a natural benchmark against which to interpret the direct duality-based construction. The comparison reveals how the change of measure induced by the duality in the approach proposed in Section \ref{sec:RSBAM} relates to the classical Kuroda–Nagai transformation, and thereby clarifies the role of the entropy penalty in \eqref{eq:DKL:Pgamma:P}.   

\subsection{A Two-Step Kuroda–Nagai/Duality Route}\label{sec:Ph_and_EEentropy:general}
This subsection shows that the direct solution derived in Section \ref{sec:RSBAM:energy-entropy} is equivalent to a sequential two-step procedure in which one first performs the Kuroda–Nagai change of measure and then applies the free energy–entropy duality under the transformed measure. Consistently with the derivations in subsections \ref{sec:energy-entropy:DS}-\ref{S.2.3.3}, and in light of the results proved in Section \ref{S.2.3.3}, we will focus here on classes of Markov controls. 

\subsubsection{Kuroda--Nagai Change of Measure from $\mathbb{P}$ to $\mathbb{P}^H$}\label{sec:Ph_and_EEentropy:KurodaNagaichange}

We start by expressing the risk-sensitive benchmarked criterion in  \eqref{eq:J:general} as follows: 
\begin{align}\label{eq:J:twostep:2}
J(t,x;H,\theta)
&:= -\frac{1}{\theta} \ln \mathbf{E}_{t,x} \left[ e^{-\theta R_T} \right]
= -\frac{1}{\theta} \ln \mathbf{E}_{t,x} \left[    
\exp \left\{ \theta \int_t^T g_1(s,X_s,h_s;\theta) ds \right\} \chi^H_{[t,T]}
\right],
\end{align}
where
\begin{align}\label{eq:g:twostep}
g_1(s,x,h;\theta)
&:= \frac{\theta+1}{2}h'\Sigma_s\Sigma_s'h 
- h' \left( a_s + A_s x \right)
- \theta h'\Sigma_s\Xi_s
+ \left(c_s + C_s x\right)
+ \frac{\theta -1}{2}\Xi_s'\Xi_s,
\end{align}
and the Dol\'eans-Dade exponential $\chi^H_{[t,T]}$ is defined via  
\begin{align}\label{eq:expmart:twostep}
    \chi^H_{[t,T]} 
    &:= \exp \bigg\{ -\theta \int_t^T\left(h_s' \Sigma_s - \Xi_s' \right) dW_s  
        -\frac{1}{2} \theta^2     \int_t^T \left(h_s' \Sigma_s - \Xi_s' \right)
        \left(\Sigma_s' h_s - \Xi_s \right)ds 
    \bigg\}.
\end{align}

Here, we need to assume that the investor's strategy $H = \left(h_s\right)_{s \in [t,T]}$ is in class $\mathcal{A}^{{HN}}$ defined below.

\begin{definition}\label{def:classAT:KN:fullobs:twostep}

$H = \left(h_s\right)_{s \in [t,T]}$ with values in $\mathbb{R}^{m}$ is in class $\mathcal{A}^{{HN}}$ if the control $h_s$ satisfies the following conditions:

\begin{enumerate}[(i)]
    \item  $h_s$ is a Markov control;

    \item $P\left(\int_t^T \left| h_s \right|^2 ds < +\infty \right) = 1$;

    \item the Dol\'eans exponential $\chi^H_{[t,T]}$ defined as \eqref{eq:expmart:twostep} is an exponential martingale.
\end{enumerate}

\end{definition}

For $H \in \mathcal{A}^{{HN}}$, let $\mathbb{P}^H$ be the probability measure on $(\Omega,\mathcal{F}_T)$ defined via the Radon-Nikodym derivative
\begin{align}\label{eq:Ph:twostep}
    \frac{d\mathbb{P}^H}{d\mathbb{P}}\Big{|}_{\mathcal{F}_T}:= \chi^H_{[t,T]}.
\end{align}
Moreover,
\begin{align}\label{eq:def:Wh}
    W^H_s = W_s +\theta \int_t^s \left(\Sigma_u' h_u - \Xi_u \right) du    
\end{align}
is a standard Wiener process under $\mathbb{P}^H$. The $\mathbb{P}^H$-dynamics of the state process $X_s$ is 
\begin{align}\label{eq:state:Ph:FO:twostep}
    d X_s
    = \left[b_s + B_s X_s -\theta\Lambda_s \left( \Sigma_s' h_s - \Xi_s\right)\right] ds
                + \Lambda_s dW^H_s
        , \;  s \in [t,T].
\end{align}
Finally, the $\mathbb{P}^H$-risk-sensitive criterion $J^H$ and exponentially-transformed criterion $I^H$ are
\begin{align}
    J^H(t,x;H,\theta) &:= -\frac{1}{\theta} \ln \mathbf{E}_{t,x}^H \left[    
\exp \left\{ \theta \int_t^T g_1(s,X_s,h_s;\theta) ds \right\}
\right]        \label{eq:J:Ph:FO} \\
    I^H(t,x;H,\theta) 
        &:= e^{-\theta J^H(t,x;H,\theta)}
        = \mathbf{E}_{t,x}^H \left[    
    \exp \left\{ \theta \int_t^T g_1(s,X_s,h_s;\theta) ds \right\} 
    \right],     \label{eq:I:Ph:FO:twostep}
\end{align}
where $\mathbf{E}_{t,x}^H \left[ \cdot \right]$ denotes the expectation taken with respect to the probability measure $\mathbb{P}^H$ and with initial condition $(t,X_t=x)$.

\begin{remark}\label{rk:nonstandard}
Under the probability measure $\mathbb{P}$, the benchmarked investment management problem is not a standard (risk-sensitive) stochastic control problem due to two important features. First, the state dynamics at \eqref{eq:state} is uncontrolled. Second, the terminal reward  $-\theta R_T$ in the exponential contains an It\^o integral with respect to the $\mathbb{P}$-Brownian motion $W_s$. The main interest of the Kuroda--Nagai change of measure is to transform this non-standard control problem into a standard linear exponential-of-quadratic Gaussian (LEQG) control problem with a controlled state process under the probability measure $\mathbb{P}^H$ \citep[see][]{kuna02,dall_RSBench,DavisLleoBook2014}.     
\end{remark}

\subsubsection{Applying the Free Energy–Entropy Duality Under $\mathbb{P}^H$}\label{sec:Ph_and_EEentropy:EEentropy}

To apply the free energy--entropy argument, we introduce the probability measure $\mathbb{P}^N$ defined from the probability measure $\mathbb{P}^H$ via the Radon-Nikodym derivative
\begin{align}\label{eq:RNderivative:nu:twostep}
    \frac{d\mathbb{P}^N}{d\mathbb{P}^H}\Big{|}_{\mathcal{F}_T} = \exp \left\{ -\frac{1}{2} \int_t^T  \| \nu_s \|^2 ds + \int_t^T \nu_s' dW^H_s \right\} =:  \chi^N_{[t,T]},
\end{align} 
for a strategy $N$ in class $\mathcal{A}^N$ defined as:

\begin{definition}\label{def:classAgammaT:fullobs:twostep}

A \emph{strategy} $N = \left(\nu_s\right)_{s \in [t,T]}$ with values in $\mathbb{R}^{d}$ is in class $\mathcal{A}^N$ if the \emph{control} $\nu_s$ satisfies the following conditions:

\begin{enumerate}[(i)]
    \item  $\nu_s$ is a Markov control;

    \item $P^H\left(\int_t^T \left| \nu_s \right|^2 ds < +\infty \right) = 1$;
    
    \item $ \chi^N_{[t,T]}$, defined at \eqref{eq:RNderivative:nu:twostep}, is an exponential martingale.

\end{enumerate}

\end{definition}

The $\mathbb{P}^N$-standard Wiener process is
\begin{align}
    W^N_s = W^H_s - \int_t^s \nu_u du
\end{align}
and the state process $X_s$ satisfies
\begin{align}\label{eq:state:Pgamma:twostep}
    d X_s
    = \left\{b_s + B_s X_s -\Lambda_s \left[\theta\left( \Sigma_s' h_s - \Xi_s \right)- \nu_s\right] \right\}ds
    + \Lambda_s dW^N_s
        , \;  s \in [t,T]. 
\end{align}

The \emph{free energy--entropy duality relation} establishes the following relation between the $\mathbb{P}^H$-criterion $I^H$ introduced at \eqref{eq:I:Ph:FO:twostep} and the $\mathbb{P}^N$-expectation of the function $g_1$ plus a regularization penalty:
\begin{align}\label{eq:RS:duality:I:twostep}
    \ln I^H(t,x;H,\theta) = \sup_{N \in \mathcal{A}^N} \mathbf{E}_{t,x}^N \left[  \int_t^T \left\{ \theta g_1(s,X_s,h_s;\theta) - \frac{1}{2} \| \nu_s \|^2 \right\}ds 
    \right].
\end{align}

As in Section \ref{sec:RSBAM:energy-entropy}, we focus on the case $\theta >0$. By Lemma \ref{lem:infexp:app} in Appendix \ref{app:Lemma:Meneghini},
\begin{align}\label{eq:RS:duality:I:thetapos:twostep}
    \inf_{H \in \mathcal{A}^{HN}} I^H(t,x;H,\theta) 
    &= \inf_{H \in \mathcal{A}^{{HN}}}\exp \left\{ \sup_{N \in \mathcal{A}^N} \mathbf{E}_{t,x}^N \left[  \int_t^T \left\{ \theta g_1(s,X_s,h_s;\theta) - \frac{1}{2} \| \nu_s \|^2 \right\} ds 
    \right]\right\}
    \nonumber\\
    &= \exp \left\{ \inf_{H \in \mathcal{A}^{{HN}}} \sup_{N \in \mathcal{A}^N} \mathbf{E}_{t,x}^N \left[ \int_t^T \left\{ \theta g_1(s,X_s,h_s;\theta) - \frac{1}{2} \| \nu_s \|^2 \right\}ds 
    \right]\right\}, 
\end{align}
where the state process $X$ is given at \eqref{eq:state:Pgamma:twostep}.  Using relation \eqref{eq:RS:duality:I:twostep}, we express $J^H(t,x;H,\theta)$ as
\begin{align}\label{eq:RS:duality:J:twostep}
    J^H(t,x;H,\theta) 
    = -\frac{1}{\theta} \ln I^H(t,x;H,\theta)
    = \inf_{N \in \mathcal{A}^N} \mathbf{E}_{t,x}^N \left[ 
    \int_t^T \left\{-g_1(s,X_s,h_s;\theta) 
    +\frac{1}{2\theta} \| \nu_s \|^2 \right\} ds 
    \right].
\end{align}
Hence,
\begin{align}\label{eq:RS:duality:I:thetapos:alt:twostep}
    \sup_{H \in \mathcal{A}^{{HN}}} J^H(t,x;H,\theta) 
    &= \sup_{H \in \mathcal{A}^{{HN}}} \inf_{N \in \mathcal{A}^N} \mathbf{E}_{t,x}^N \left[ 
    \int_t^T \left\{- g_1(s,X_s,h_s;\theta) 
    +\frac{1}{2\theta} \| \nu_s \|^2 \right\} ds 
    \right].    
\end{align}

We have therefore arrived again at a stochastic differential game formulation, now under the probability measure $\mathbb P^N$, in which the state process is controlled jointly by $h_s$ and $\nu_s$. As in subsection \ref{S.2.3.3}, we shall obtain optimal strategies defined by policies that are affine functions of the state variable $X$.

We may then define a value function $U(t,x)$ as
\begin{align}\label{eq:valuefunction:U}
    U(t,x) := \sup_{H \in \mathcal{A}^{{HN}}} \inf_{N \in \mathcal{A}^N} \mathbf{E}_{t,x}^N \left[ 
    \int_{t}^{T} \left\{- g_1(s,X_s,h_s;\theta) 
    +\frac{1}{2\theta} \| \nu_s \|^2 \right\} ds 
    \right],
\end{align}
to which we associate the Bellman--Isaacs PDE 
\begin{align}\label{eq:BellmanIsaacs:twostep}
   \frac{\partial U(s,x)}{\partial{s}} + \mathcal{H}^+(s,x,DU(s,x),D^2U(s,x)) = 0,
   \qquad U(T,x) = 0, \forall x \in \mathbb{R}^n,
\end{align}
where $DU$ and $D^2U$ are defined as in Section \ref{S.2.3.3}, and
\begin{align}\label{eq:H:twostep}
   & \mathcal{H}^+(s,x,p,M) 
    \nonumber\\
   :=& \sup_{h \in \mathbb{R}^m}\inf_{\nu \in \mathbb{R}^d}
   \Bigg\{ \left\{b_s + B_s x -\Lambda_s \left[\theta\left( \Sigma_s' h - \Xi_s\right) - \nu\right]\right\}' p 
    + \frac{1}{2} \tr \left(\Lambda_s\Lambda_s' M \right)
    - g_1(s,x,h;\theta)
    + \frac{1}{2\theta}\nu'\nu
    \Bigg\}
\end{align}
for $p \in \mathbb{R}^n,$ $M \in \mathbb{R}^{n \times n}$. 

Proceeding as in the proof of Proposition \ref{prop:candidate_controls}, we obtain the candidate policies $\hat{h}(s,x,DU(s,x))$ and $\hat{\nu}(s,x,DU(s,x))$ as
\begin{align}\label{eq:candidate_h_nu}
\begin{cases}
    \hat{h}(s,x,DU(s,x))
    = \frac{1}{\theta+1} \left(\Sigma_s\Sigma_s'\right)^{-1}
      \left[a_s + A_s x + \theta\Sigma_s\Xi_s - \theta\Sigma_s\Lambda_s'DU(s,x) \right], \\
    \hat{\nu}(s,x,DU(s,x))
    = -\theta \Lambda_s'DU(s,x).
\end{cases}
\end{align}
Then, we define the candidate strategies $\hat{H}:= \left(\hat{h}_s = \hat{h}(s,X_s,DU(s,X_s))\right)_{s \in [t,T]}$ and $\hat{N}:= \left(\hat{\nu}_s = \hat{\nu}(s,X_s,DU(s,X_s))\right)_{s \in [t,T]}$.

Since the Hamiltonian in \eqref{eq:H:twostep} is strictly concave in $h$ and strictly convex in $\nu$, these candidate policies define the pointwise maximizer and minimizer, respectively. Moreover, we can check that
\begin{align}
    & \mathcal{H}^+(s,x,p, M) 
                \nonumber\\
    &= \inf_{\nu \in \mathbb{R}^d}\sup_{h \in \mathbb{R}^m}
   \Bigg\{ \left\{b_s + B_s x -\Lambda_s \left[\theta\left( \Sigma_s' h - \Xi_s\right) - \nu\right]\right\}' p 
    + \frac{1}{2} \tr \left(\Lambda_s\Lambda_s' M \right)
    - g_1(s,x,h;\theta)
    + \frac{1}{2\theta}\nu'\nu
    \Bigg\}
                \nonumber\\
    &=: \mathcal{H}^-(s,x,p, M),
\end{align}
along the same lines as in the Proof of Proposition \ref{prop:candidate_controls}. Therefore, the Bellman--Isaacs minimax condition is satisfied, and $U$ is well defined as the value function for the stochastic differential game. In light of the structure of risk-sensitive control problems, the following Lemma and  Proposition show, respectively, that the value function $U$ is quadratic in the state, and that the candidate policies are optimal.

\begin{lemma}\label{lem:twosteps:sol_IsaacPDE}
    The value function $U(t,x)$ is of the form 
    \begin{align}\label{eq:sol:U}
        U(t,x) = \frac{1}{2} x' Q_t x + q_t'x + k_t,
    \end{align}
    where $Q: [t,T] \to \mathbb{R}^{n \times n}, q : [t,T] \to \mathbb{R}^n$, and $k: [t,T] \to \mathbb{R}$ solve \eqref{eq:Q:Riccati}, \eqref{eq:q:ODE}, and \eqref{eq:k:integral} respectively.
\end{lemma}

\begin{proof}[Sketch of proof]
    The proof follows along the lines of that of Theorem \ref{theo:sol_IsaacPDE}. Substitute the candidate policies at \eqref{eq:candidate_h_nu} and the quadratic form \eqref{eq:sol:U} into \eqref{eq:BellmanIsaacs:twostep}. Reorganize as a quadratic form in $x$ to recover equations \eqref{eq:Q:Riccati}, \eqref{eq:q:ODE}, and \eqref{eq:k:integral} for $Q_t$, $q_t$, and $k_t$.  

\end{proof}

\begin{proposition}[Optimality of the Candidate Policies]\label{prop:optimality:policies}
\;
Let $U$ be the value function defined at \eqref{eq:valuefunction:U} and having the quadratic expression in \eqref{eq:sol:U}. Let $H^* = \left(h^*(s,X_s) \right)_{s \in [t,T]} := \left(\hat{h}(s,X_s,DU(s,X_s)) \right)_{s \in [t,T]}$ and $N^* = \left(\nu^*(s,X_s) \right)_{s \in [t,T]}  := \left(\hat{\nu}(s,X_s,DU(s,X_s)) \right)_{s \in [t,T]}$, where $\hat{h}(\cdot)$ and $\hat{\nu}(\cdot)$ are the candidate policies given at \eqref{eq:candidate_h_nu} . Then, the strategies $H^*$ and $N^*$ are admissible according to Definitions \ref{def:classAT:KN:fullobs:twostep} and \ref{def:classAgammaT:fullobs:twostep} respectively and are optimal.
\end{proposition}

\begin{proof}[Sketch of proof]  
The proof follows along the lines of that of Theorem \ref{theo:verification}. Showing that $H^*$ is admissible also requires showing that $ \chi^{H^*}_{[t,T]}$, defined at \eqref{eq:expmart:twostep} for $h=h^*$, is an exponential martingale. To see this, note that $U$ is quadratic in $x$ by  Lemma \ref{lem:twosteps:sol_IsaacPDE}, so $\hat{h}(s,x,DU(s,x))$ is affine in $x$. Moreover, under the candidate strategy $H^*$, the state process $X$ with $\mathbb{P}^H$ dynamics \eqref{eq:state:Ph:FO:twostep} is Gaussian. Since $h^*(s,X_s)$ is affine in the Gaussian state $X_s$, it is square-integrable on $[t,T]$. Under the standing $C^1$ assumptions on the deterministic coefficients and the finite horizon, the coefficients are bounded and the Novikov condition holds; therefore, $\chi^{H^*}_{[t,T]}$ is a true martingale. Hence, condition (iii) in Definition \ref{def:classAT:KN:fullobs:twostep} is satisfied. It remains to show that $N^*$ is admissible in the sense of Definition \ref{def:classAgammaT:fullobs:twostep}, and specifically, that $\chi^{N^*}_{[t,T]}$, defined at \eqref{eq:RNderivative:nu:twostep} for $\nu=\nu^*$, is an exponential martingale. This is shown in full analogy to the proof of Theorem \ref{theo:verification}, in particular, where it is shown that $\Gamma^*\in \mathcal{A}^{\Gamma}$.   

\end{proof}  

The following theorem shows that the sequential route consisting of the Kuroda–Nagai change of measure followed by the free energy–entropy duality yields exactly the same value function and optimal investment strategy as the direct one-step approach of Sections \ref{sec:energy-entropy:DS}–\ref{S.2.3.3}.

\begin{theorem}\label{theo:twostep_and_onestep:soloptimcontrol}
\hfill
\begin{enumerate}[(i)]
    \item The value functions $u$ defined at \eqref{eq:valuefunction:u} and $U$ defined at \eqref{eq:valuefunction:U} are related by
    \begin{align}\label{eq:U_so_u}
        u(t,x) = -\theta U(t,x);
    \end{align}
    
    \item The definitions of the asset allocation $\hat{h}(s,x,Du(s,x))$ at \eqref{eq:hhat}
    and $\hat{h}(s,x,DU(s,x))$ at \eqref{eq:candidate_h_nu} are equivalent, that is,
    \begin{align}
        \hat{h}(s,x,Du(s,x))
        =
        \hat{h}(s,x,DU(s,x))
        =: \hat{h}(s,x),
        \quad \forall (s,x)\in[t,T]\times\mathbb{R}^n.
    \end{align}
    Hence, the strategy $H^*:=\left(\hat{h}(s,X_s)\right)_{s\in[t,T]}$ is optimal for the original problem \eqref{eq:valuefunction:u} and the two-step problem \eqref{eq:valuefunction:U}.
    
    \item The policies $\hat{\gamma}$ defined at \eqref{eq:gammahat:1} and $\hat{\nu}$ defined at \eqref{eq:candidate_h_nu} are related by
\begin{align}
    \hat{\gamma}(s,x,Du(s,x))
    =
    \hat{\nu}(s,x,DU(s,x))
    -\theta\left(\Sigma_s'\hat{h}(s,x,DU(s,x))-\Xi_s\right),
    \quad \forall (s,x)\in[t,T]\times\mathbb{R}^n.
\end{align}
\end{enumerate}
\end{theorem}

\begin{proof}
\hfill
\begin{enumerate}[(i)]
    \item     
    From Theorem \ref{theo:sol_IsaacPDE} and Lemma \ref{lem:twosteps:sol_IsaacPDE}, 
    \begin{align}
        u(t,x) 
        = -\theta \left( \frac{1}{2} x' Q_t x + q_t'x + k_t \right) 
        = -\theta U(t,x) 
    \end{align}
    where $Q: [t,T] \to \mathbb{R}^{n \times n}, q : [t,T] \to \mathbb{R}^n$, and $k: [t,T] \to \mathbb{R}$ solve \eqref{eq:Q:Riccati}, \eqref{eq:q:ODE}, and \eqref{eq:k:integral} respectively.

    \item From \eqref{eq:candidate_h_nu},
    \begin{align}
        \hat{h}(s,x,DU(s,x)) 
        =&  \frac{1}{\theta+1} \left(\Sigma_s\Sigma_s'\right)^{-1} \left[a_s + A_s x + \theta\Sigma_s\Xi_s - \theta\Sigma_s\Lambda_s'DU(s,x) \right]     
                                    \nonumber\\
        \underset{\text{by \eqref{eq:U_so_u}}}{=}&  \frac{1}{\theta+1} \left(\Sigma_s\Sigma_s'\right)^{-1} \left[a_s + A_s x + \theta\Sigma_s\Xi_s + \Sigma_s\Lambda_s'Du(s,x) \right],
    \end{align}
    which is \eqref{eq:hhat}. Therefore, the strategy $H^* := \left(\hat{h}(s,X_s,Du(s,X_s))\right)_{s \in [t,T]} =  \left(\hat{h}(s,X_s,DU(s,X_s))\right)_{s \in [t,T]}$ is optimal for the original problem \eqref{eq:valuefunction:u} and the two-step problem \eqref{eq:valuefunction:U}.

    \item Start from \eqref{eq:gammahat:1}:
    \begin{align}
        \hat{\gamma}(s,x,Du(s,x))
        =&  
        -\theta\Lambda_s'\left(Q_s x + q_s\right) 
                                \nonumber\\
        &- \theta \left[\frac{1}{\theta+1} \Sigma_s'\left(\Sigma_s\Sigma_s'\right)^{-1}
        \left( a_s + A_s x + \theta \Sigma_s\Xi_s - \theta \Sigma_s\Lambda_s' (Q_s x+q_s) \right) - \Xi_s \right]
                                \nonumber\\
        \underset{\text{by \eqref{eq:candidate_h_nu} and Lemma \ref{lem:twosteps:sol_IsaacPDE}}}{=}&
            \hat{\nu}(s,x,DU(s,x))
            - \theta \left(\Sigma_s'\hat{h}(s,x,DU(s,x)) - \Xi_s \right).
    \end{align}
\end{enumerate}

\end{proof}

\subsection{Analysis of the Change of Measure in the Two-Step Approach}\label{sec:energy-entropy:Discussion:MeasureChange}

The candidate policy $\hat{\gamma}(s,x,Du(s,x))$ defined at \eqref{eq:gammahat:1} induces a change of measure from $\mathbb{P}$ to a probability measure $\mathbb{P}^{\hat{\Gamma}}$ via the Radon-Nikodym derivative at \eqref{eq:RNderivative:gamma}.
Proposition \ref{prop:twostep_and_onestep:measure:decomposition} below decomposes this change of measure into two sequential changes of measure. The first is the Kuroda--Nagai change of measure introduced at \eqref{eq:Ph:twostep}. The second is equivalent to the change of measure \eqref{eq:RNderivative:nu:twostep} induced by the duality under $\mathbb{P}^H$, as shown in Corollary \ref{coro:twostep_and_onestep:measure:decomposition} below.

\begin{proposition}\label{prop:twostep_and_onestep:measure:decomposition}
For the pair of candidate strategies $\left(\hat{H},\hat{\Gamma}\right)$, the Radon-Nikodym derivative $\frac{d\mathbb{P}^{\hat{\Gamma}}}{d\mathbb{P}}\big{|}_{\mathcal{F}_T}$ admits the following decomposition:
\begin{align}
        \frac{d\mathbb{P}^{\hat{\Gamma}}}{d\mathbb{P}}\Big{|}_{\mathcal{F}_T}
        = \frac{d \mathbb{P}^{\hat H}}{d \mathbb{P}} \Big{|}_{\mathcal{F}_T}
        \times \frac{d \mathbb{P}^{\hat{\Gamma}}}{d \mathbb{P}^{\hat H}}\Big{|}_{\mathcal{F}_T},
\end{align}
where  (i) $\frac{d\mathbb{P}^{\hat{\Gamma}}}{d\mathbb{P}}\big{|}_{\mathcal{F}_T}$ is, for $\Gamma=\hat\Gamma$, defined at \eqref{eq:RNderivative:gamma} and \eqref{eq:RNderivative:gammahat}, (ii)  $\frac{d \mathbb{P}^{\hat H}}{d \mathbb{P}}\big{|}_{\mathcal{F}_T}$ is defined at \eqref{eq:Ph:twostep} via the Dol\'eans-Dade exponential $\chi^H_{[t,T]}$  at \eqref{eq:expmart:twostep}, with $H$ evaluated at $\hat{H}$, and (iii) $\frac{d \mathbb{P}^{\hat{\Gamma}}}{d \mathbb{P}^{\hat H}}\big{|}_{\mathcal{F}_T}$ is defined as
\begin{align}\label{eq:expmart:chi_h-to-gamma:discussion}
    \frac{d \mathbb{P}^{\hat{\Gamma}}}{d \mathbb{P}^{\hat H}}\Big{|}_{\mathcal{F}_T}  
    &= \exp \bigg\{ \int_t^T Du(s,X_s)'\Lambda_s dW^{\hat{H}}_s  
    -\frac{1}{2} \int_t^T \| \Lambda_s' Du(s,X_s)\|^2 ds \bigg\}.
\end{align} 
\end{proposition}

\begin{proof}
Start from the definition of $\frac{d\mathbb{P}^{\hat{\Gamma}}}{d\mathbb{P}}\big{|}_{\mathcal{F}_T}$ at \eqref{eq:RNderivative:gammahat} and use the compact notation $\hat{h}_s := \hat{h}(s,X_s,DU(s,X_s)), s\in [t,T]$ introduced in Theorem \ref{theo:twostep_and_onestep:soloptimcontrol} (ii):

\begin{align}
    &\frac{d\mathbb{P}^{\hat{\Gamma}}}{d\mathbb{P}}\Big{|}_{\mathcal{F}_T}
                                    \nonumber\\
    =& \exp \bigg\{  
        -\frac{1}{2} \theta^2     \int_t^T 
        \| \Sigma_s' \hat{h}_s - \Xi_s \|^2 ds
        -\frac{1}{2} \int_t^T \| \Lambda_s' Du(s,X_s)\|^2 ds
        + \theta \int_t^T Du(s,X_s)'\Lambda_s \left( \Sigma_s' \hat{h}_s - \Xi_s\right) ds
                                        \nonumber\\
    \phantom{=}&  
        + \int_t^T Du(s,X_s)'\Lambda_s dW_s  
        -\theta \int_t^T\left(\hat{h}'(s) \Sigma_s - \Xi_s' \right) dW_s \bigg\}
                                        \nonumber\\
    =& \exp \bigg\{  
        -\theta \int_t^T\left(\hat{h}'(s) \Sigma_s - \Xi_s' \right) dW_s  
        -\frac{1}{2} \theta^2     \int_t^T 
        \| \Sigma_s' \hat{h}_s - \Xi_s \|^2 ds
                                        \nonumber\\
    \phantom{=}&  
        +  \int_t^T Du(s,X_s)'\Lambda_s \underbrace{\left[ \theta \left( \Sigma_s' \hat{h}_s - \Xi_s\right) ds + dW_s \right]}_{=dW^{\hat{H}}_s \mathrm{by \eqref{eq:def:Wh}}}  \bigg\}
        -\frac{1}{2} \int_t^T \| \Lambda_s' Du(s,X_s)\|^2 ds
                                        \nonumber\\
    =& \frac{d \mathbb{P}^{\hat H}}{d \mathbb{P}}\Big{|}_{\mathcal{F}_T} 
        \times \exp \bigg\{ \int_t^T Du(s,X_s)'\Lambda_s dW^{\hat{H}}_s  
    -\frac{1}{2} \int_t^T \| \Lambda_s' Du(s,X_s)\|^2 ds \bigg\}
\end{align}

Since $\frac{d\mathbb{P}^{\hat{\Gamma}}}{d\mathbb{P}}\big{|}_{\mathcal{F}_T}= \frac{d \mathbb{P}^{\hat H}}{d \mathbb{P}}\big{|}_{\mathcal{F}_T} \times \frac{d \mathbb{P}^{\hat{\Gamma}}}{d \mathbb{P}^{\hat H}}\big{|}_{\mathcal{F}_T}$, we conclude that $\frac{d \mathbb{P}^{\hat{\Gamma}}}{d \mathbb{P}^{\hat H}}\Big{|}_{\mathcal{F}_T}
= \exp \bigg\{ \int_t^T Du(s,X_s)'\Lambda_s dW^{\hat H}_s
-\frac{1}{2} \int_t^T \| \Lambda_s' Du(s,X_s)\|^2 ds \bigg\}$.

\end{proof}

\begin{corollary}\label{coro:twostep_and_onestep:measure:decomposition}
   At the pair of candidate strategies $\left(\hat{H},\hat{\Gamma}\right)$ defined via \eqref{eq:hhat} and \eqref{eq:gammahat:2}, the change of measure induced by the Radon-Nikodym derivatives $\frac{d\mathbb{P}^{\hat{N}}}{d\mathbb{P}^{\hat H}}\big{|}_{\mathcal{F}_T}$, defined according to \eqref{eq:RNderivative:nu:twostep}, and 
    $\frac{d \mathbb{P}^{\hat{\Gamma}}}{d \mathbb{P}^{\hat H}}\big{|}_{\mathcal{F}_T}$, given at \eqref{eq:expmart:chi_h-to-gamma:discussion}, induce the same probability measure on $(\Omega, \mathcal{F}_T)$, that is $\mathbb{P}^{\hat{N}} = \mathbb{P}^{\hat{\Gamma}}$.
\end{corollary}

\begin{proof}
Start from \eqref{eq:expmart:chi_h-to-gamma:discussion}:
\begin{align}
\frac{d \mathbb{P}^{\hat{\Gamma}}}{d \mathbb{P}^{\hat H}}\Big{|}_{\mathcal{F}_T}
    &= \exp \bigg\{ \int_t^T Du(s,X_s)'\Lambda_s dW^{\hat{H}}_s  
    -\frac{1}{2} \int_t^T \| \Lambda_s' Du(s,X_s)\|^2 ds \bigg\}
                                            \nonumber\\
    &\underset{\text{by Theo. \ref{theo:sol_IsaacPDE}}}{=} \exp \bigg\{ \int_t^T -\theta (Q_s X_s + q_s)'\Lambda_s dW^{\hat{H}}_s  
    -\frac{1}{2} \int_t^T \| -\theta \Lambda_s'(Q_s X_s + q_s) \|^2 ds \bigg\}    
                                            \nonumber\\
    &\underset{\text{by \eqref{eq:sol:U}}}{=} \exp \bigg\{ \int_t^T - \theta  DU(s,X_s)' \Lambda_s dW^{\hat{H}}_s  
    -\frac{1}{2} \int_t^T \| -\theta \Lambda_s' DU(s,X_s) \|^2 ds \bigg\}
                                            \nonumber\\
    &\underset{\text{by \eqref{eq:candidate_h_nu}}}{=} \exp \bigg\{ \int_t^T \hat{\nu}(s,X_s,DU(s,X_s))' dW^{\hat{H}}_s  
    -\frac{1}{2} \int_t^T \| \hat{\nu}(s,X_s,DU(s,X_s)) \|^2 ds \bigg\}.
\end{align}
The conclusion follows from the definition of $\frac{d\mathbb{P}^{\hat{N}}}{d\mathbb{P}^{\hat H}}\big{|}_{\mathcal{F}_T}$ via \eqref{eq:RNderivative:nu:twostep} evaluated at $N = \hat{N}$.

\end{proof}

The first change of measure, via the Radon-Nikodym derivative $\frac{d \mathbb{P}^{\hat H}}{d \mathbb{P}}\big{|}_{\mathcal{F}_T}$ is used to rewrite the non-standard risk-sensitive benchmarked investment management problem into an LEQG problem, as shown in Section \ref{sec:Ph_and_EEentropy:KurodaNagaichange}. 

The second change of measure, via the Radon-Nikodym derivative
$\frac{d \mathbb{P}^{\hat{\Gamma}}}{d \mathbb{P}^{\hat H}}\Big{|}_{\mathcal{F}_T}
=
\frac{d \mathbb{P}^{\hat{N}}}{d \mathbb{P}^{\hat H}}\Big{|}_{\mathcal{F}_T}$,
is specific to the free energy--entropy duality, which recasts the LEQG control problem under the probability measure $\mathbb{P}^{\hat H}$ into an LQG game under the probability measure $\mathbb{P}^{\hat{\Gamma}}$. From a PDE perspective, plugging $\hat{\nu}_s$ into the Bellman--Isaacs equation \eqref{eq:BellmanIsaacs:twostep} yields the quadratic gradient term
$-\frac{\theta}{2}\,DU(s,x)'\Lambda_s\Lambda_s'DU(s,x)$,
which characterizes the HJB PDE for the LEQG control problem \citep[see the last paragraph of Section 2.4.2 in][]{DavisLleoBook2014}. Consequently, at $\nu_s = \hat{\nu}_s$, the Bellman--Isaacs equation for the stochastic differential game reduces to the HJB PDE for the LEQG stochastic control problem. This interpretation also explains the observation in Remark \ref{rk:gammahat:nonlinearterm}.

Performing the two changes of measure sequentially leads to a surprisingly simple situation. The terms inside the $\sup\inf$ of the Hamiltonian $\mathcal{H}^+$ at \eqref{eq:H:twostep}  are still quadratic in the controls, but without interaction, that is, without a term of the form $h'M\nu$ for some matrix $M$. Why did we get this result? The control $h_s$ enters the dynamics of the state process $X$ through the first change of measure from $\mathbb{P}$ to $\mathbb{P}^H$. So, $h_s$ only affects the drift of $X$, not its diffusion. Consequently, in the second change of measure, the additional drift component depends only on $\nu_s$, so no bilinear interaction term between $h_s$ and $\nu_s$ appears in the Hamiltonian.

\begin{corollary}[Special Case of the Kelly Portfolio]\label{coro:Kelly:2}
Taking the limit as $\theta \to 0$ yields the Kelly or growth-optimal portfolio. In this case, notice that the probability measure $\mathbb{P}^{\hat H}$ is identical to the probability measure $\mathbb{P}$, but the Radon–Nikodym derivative defining $\mathbb{P}^{\hat\Gamma}$ remains of the same form. Hence, 
\begin{align}\label{eq:expmart:chi_h-to-gamma:discussion:2}
    \frac{d \mathbb{P}^{\hat{\Gamma}}}{d \mathbb{P}^{\hat H}} \Big{|}_{\mathcal{F}_T}
        &= \exp \bigg\{ \int_t^T Du(s,X_s)'\Lambda_s dW^{\hat{H}}_s  
        -\frac{1}{2} \int_t^T \| \Lambda_s' Du(s,X_s)\|^2 ds \bigg\} = \frac{d \mathbb{P}^{\hat{\Gamma}}}{d \mathbb{P}}\Big{|}_{\mathcal{F}_T},
\end{align}
consistently with Remark \ref{rk:Kelly:1}.
\end{corollary}

\section{Implementation}\label{sec:Implementation}

This section provides a simple computational illustration of the model using market data. We consider a U.S. equity portfolio manager with a five-year investment horizon and risk sensitivity parameter $\theta = 1$. The benchmark allocates 90\% to the S\&P 500, which comprises the 500 largest publicly traded firms by market capitalization, and 10\% to the S\&P 400, which comprises 400 mid-capitalization firms. The manager uses a six-factor model consisting of the five Fama--French factors \citep{FF2015} together with Momentum \citep{Carhart1997}.

The investment universe contains $m = 13$ U.S. exchange-traded funds (ETFs). The first 11 ETFs track the S\&P 500 sectors defined by the Global Industry Classification Standard (GICS), so the S\&P 500 can be replicated by holding these ETFs at their sector weights. The twelfth ETF is the iShares Core S\&P 400 Mid-Cap ETF (ticker: IJH), which tracks 400 mid-cap stocks. The thirteenth ETF is the iShares Core S\&P 600 Small-Cap ETF (ticker: IJR), which tracks the S\&P 600 Small-Cap Index. Since the S\&P 600 is not included in the benchmark, any allocation to small-capitalization stocks is purely tactical.

Model parameters are estimated from daily returns over June 20, 2018, to December 31, 2024, giving $\mathcal{T} = 1644$ observations and $\Delta t := \frac{1}{252}$. June 20, 2018, is the earliest date for which all series are available, because a GICS revision changed sector definitions and composition. The sample therefore spans a volatile period in the U.S. equity market, including the rise in 2019, the COVID-19 crash in March 2020, and the subsequent rally. Factor data and money market returns are obtained from Kenneth French's data library\footnote{The data are available online at \url{http://mba.tuck.dartmouth.edu/pages/faculty/ken.french/Data_Library/f-f_factors.html}.}, while ETF returns are taken from the CRSP database\footnote{©2026 Center for Research in Security Prices (CRSP), LLC.}.

Because the objective of this section is illustrative rather than predictive, we keep the implementation deliberately simple and assume constant model parameters throughout. Drift coefficients are estimated by applying standard regression methods to a discretized version of the SDE for $X_t$. Diffusion matrices are estimated from quadratic and cross variations using resampling methods. Here, the benchmark is itself a combination of the $m$ traded assets, so $L_t$ is modeled as a fixed-weight portfolio corresponding to the benchmark weights.

Using these estimates, we run 5,000 simulations over 1,260 daily time steps. This yields terminal outcomes for the portfolio and the benchmark, together with 6,300,000 simulated daily log excess returns over the money market rate, which are used to estimate return distributions and performance measures.

Figure \ref{fig:rangeoftheta:excessreturn} reports the distribution of daily log excess returns over the money market rate for the benchmark, the Kelly portfolio, the portfolio obtained via the Kuroda--Nagai change-of-measure approach, and the portfolio obtained via the Free Energy--Entropy Duality. Table \ref{table:rangeoftheta:excessreturn} reports summary statistics, tail-risk measures, and risk-adjusted performance measures for the same portfolios. Consistent with the equivalence results established in Section \ref{sec:Kuroda--Nagai}, the Free Energy--Entropy Duality and the Kuroda--Nagai change-of-measure approach yield identical simulated performance.

\begin{figure}
\begin{center}
\includegraphics[width=0.75\textwidth]{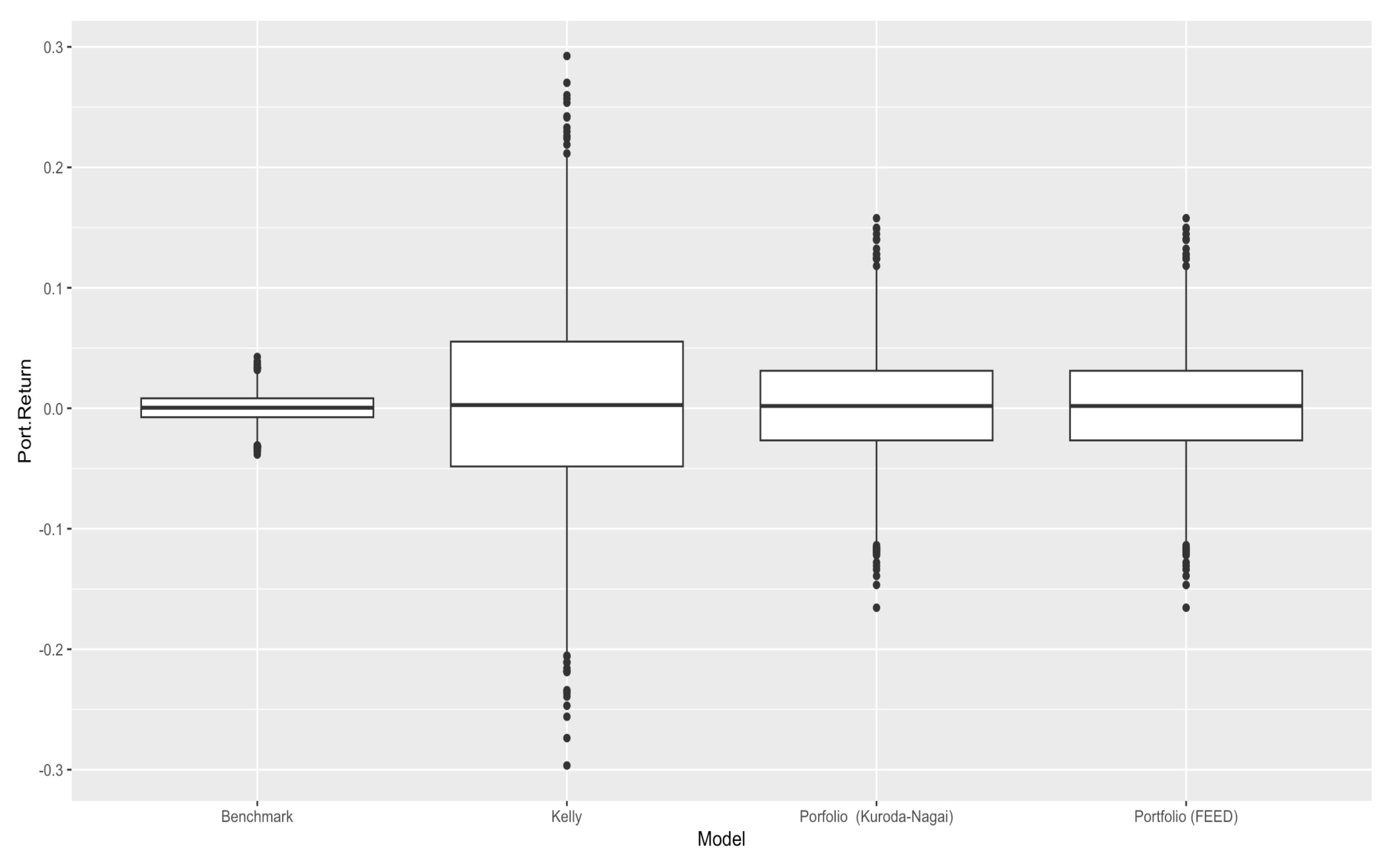}\caption{Box-and-whisker plot of the distribution of daily log excess returns over the money market rate for the benchmark, the Kelly portfolio, the portfolio obtained via the Kuroda--Nagai change-of-measure approach, and the portfolio obtained via the Free Energy--Entropy Duality.}
	\label{fig:rangeoftheta:excessreturn}
    \end{center}
\end{figure}

\begin{table}
\begin{center}
\begin{tabular}{l rrrr}
\hline
 & Benchmark & Portfolio (KN) & Portfolio (FEED) & Kelly \\
\hline\hline
Mean & 0.0507 \% & 0.2437 \% & 0.2437 \% & 0.3078 \% \\
Standard deviation & 1.1682 \% & 4.3154 \% & 4.3154 \% & 7.8217 \% \\
Semideviation & 0.8569 \% & 3.0500 \% & 3.0500 \% & 5.5284 \% \\
Skewness & -0.5895 & 0.0035 & 0.0035 & 0.0030 \\
Kurtosis & 17.4858 & 0.0401 & 0.0401 & 0.0450 \\
\\
VaR 95\%  & 1.6773 \% & 7.0890 \% & 7.0890 \% & 12.8507 \% \\
CVaR 95\% & 2.8334 \% & 8.9186 \% & 8.9186 \% & 16.1727 \% \\
\\
Sharpe & 0.0434 & 0.0565 & 0.0565 & 0.0394 \\
Sortino & 0.0607 & 0.0836 & 0.0836 & 0.0575 \\
Mean-to-VaR 95\% & 0.0302 & 0.0344 & 0.0344 & 0.0240 \\
Mean-to-CVaR 95\% & 0.0179 & 0.0273 & 0.0273 & 0.0190 \\
\hline
\end{tabular}
\end{center}
\caption{Performance measures for daily log excess returns over the money market rate for the benchmark, the portfolio obtained via the Free Energy--Entropy Duality, the portfolio obtained via the Kuroda--Nagai change-of-measure approach, and the Kelly portfolio. The reported measures include summary statistics (mean, standard deviation, semideviation, skewness, and kurtosis), tail-risk measures (VaR and CVaR at the 95\% confidence level, expressed as percentage losses relative to the mean return), and risk-adjusted performance measures (Sharpe ratio, Sortino ratio, mean-to-VaR, and mean-to-CVaR).}\label{table:rangeoftheta:excessreturn}
\end{table}

\section{Conclusion}\label{sec:Conclusion}

This paper develops a new formulation for the solution of risk-sensitive benchmarked investment management through the free energy--entropy duality. The new formulation yields a direct, explicit solution to the benchmarked risk-sensitive portfolio problem, embedding the classical Kuroda--Nagai change-of-measure approach within a more general framework. Specifically, it transforms a non-standard risk-sensitive control problem into a LQG game problem under a suitable equivalent probability measure, with an entropic regularization. When the model parameters are known or can be estimated with high accuracy, this leads to an implementation that is both tractable and implementable. In a companion paper, we relax the assumption of known model parameters and show that the present framework is naturally suited to a reinforcement learning implementation.

%
%


\begin{thebibliography}{}

\bibitem[Bensoussan, 1992]{Bensoussan1992}
Bensoussan, A. (1992).
\newblock {\em Stochastic Control of Partially Observable Systems}.
\newblock Cambridge University Press.

\bibitem[Bielecki \& Pliska, 1999]{bipl99}
Bielecki, T. \& Pliska, S. (1999).
\newblock Risk-sensitive dynamic asset management.
\newblock {\em Applied Mathematics and Optimization}, 39, 337--360.

\bibitem[Bielecki \& Pliska, 2000]{bipl00}
Bielecki, T. \& Pliska, S. (2000).
\newblock Risk sensitive asset management with transaction costs.
\newblock {\em Finance and Stochastics}, 4, 1--33.

\bibitem[Bielecki \& Pliska, 2003]{Bielecki2003}
Bielecki, T. \& Pliska, S. (2003).
\newblock Economic properties of the risk sensitive criterion for portfolio management.
\newblock {\em The Review of Accounting and Finance}, 2(2), 3--17.

\bibitem[Bielecki \& Pliska, 2004]{bipl04}
Bielecki, T. \& Pliska, S. (2004).
\newblock Risk sensitive intertemporal \textsc{CAPM}.
\newblock {\em \textsc{IEEE} Transactions on Automatic Control}, 49(3), 420--432.

\bibitem[Bielecki et~al., 2022]{bieleckiRiskSensitiveMarkovDecision2022a}
Bielecki, T.~R., Chen, T., \& Cialenco, I. (2022).
\newblock Risk-{{Sensitive Markov Decision Problems}} under {{Model Uncertainty}}: {{Finite Time Horizon Case}}.
\newblock In G. Yin \& T. Zariphopoulou (Eds.), {\em Stochastic {{Analysis}}, {{Filtering}}, and {{Stochastic Optimization}}: {{A Commemorative Volume}} to {{Honor Mark H}}. {{A}}. {{Davis}}'s {{Contributions}}}  (pp.\ 33--52). Cham: Springer International Publishing.

\bibitem[Carhart, 1997]{Carhart1997}
Carhart, M. (1997).
\newblock On persistence in mutual fund performance.
\newblock {\em The Journal of Finance}, 52(1), 57--82.

\bibitem[Dai~Pra et~al., 1996]{daipraConnectionsStochasticControl1996}
Dai~Pra, P., Meneghini, L., \& Runggaldier, W.~J. (1996).
\newblock Connections between stochastic control and dynamic games.
\newblock {\em Mathematics of Control, Signals and Systems}, 9(4), 303--326.

\bibitem[Davis \& Lleo, 2008]{dall_RSBench}
Davis, M. \& Lleo, S. (2008).
\newblock Risk-sensitive benchmarked asset management.
\newblock {\em Quantitative Finance}, 8(4), 415--426.

\bibitem[Davis \& Lleo, 2011]{dall_JDRSAM_Diff}
Davis, M. \& Lleo, S. (2011).
\newblock Jump-diffusion risk-sensitive asset management {I}: Diffusion factor model.
\newblock {\em SIAM Journal on Financial Mathematics}, 2, 22--54.

\bibitem[Davis \& Lleo, 2013a]{dall_BLcontinuous}
Davis, M. \& Lleo, S. (2013a).
\newblock Black-{L}itterman in continuous time: The case for filtering.
\newblock {\em Quantitative Finance Letters}, 1(1), 30--35.

\bibitem[Davis \& Lleo, 2013b]{dall_JDRSAM}
Davis, M. \& Lleo, S. (2013b).
\newblock Jump-diffusion risk-sensitive asset management ii: Jump-diffusion factor model.
\newblock {\em SIAM Journal on Control and Optimization}, 51(2), 1441.

\bibitem[Davis \& Lleo, 2013c]{dall_JDBench}
Davis, M. \& Lleo, S. (2013c).
\newblock Jump-diffusion risk-sensitive benchmarked asset management.
\newblock In H. Gassmann \& W. Ziemba (Eds.), {\em Stochastic Programming: Applications in Finance, Energy, Planning and Logistics}  (pp.\ 97--128).: World Scientific Publishing.

\bibitem[Davis \& Lleo, 2014]{DavisLleoBook2014}
Davis, M. \& Lleo, S. (2014).
\newblock {\em Risk-Sensitive Investment Management}, volume~19 of {\em Advanced Series on Statistical Science and Applied Probability}.
\newblock World Scientific Publishing.

\bibitem[Davis \& Lleo, 2015]{dall_JDRSALM}
Davis, M. \& Lleo, S. (2015).
\newblock Jump-diffusion asset-liability management via risk-sensitive control.
\newblock {\em OR Spectrum}, 37(3), 655--675.

\bibitem[Davis \& Lleo, 2021]{davisRisksensitiveBenchmarkedAsset2021}
Davis, M. \& Lleo, S. (2021).
\newblock Risk-sensitive benchmarked asset management with expert forecasts.
\newblock {\em Mathematical Finance}, 31(4), 1162--1189.

\bibitem[Davis \& Lleo, 2024]{dall_JDBenchAltData2024}
Davis, M. \& Lleo, S. (2024).
\newblock Jump-diffusion risk-sensitive benchmarked asset management with traditional and alternative data.
\newblock {\em Annals of Operations Research}, 336, 661--689.

\bibitem[Fama \& French, 2015]{FF2015}
Fama, E. \& French, K.~R. (2015).
\newblock A five-factor asset pricing model.
\newblock {\em Journal of Financial Economics}, 116(1), 1--22.

\bibitem[Fleming \& Soner, 2006]{flso06}
Fleming, W. \& Soner, H. (2006).
\newblock {\em Controlled Markov Processes and Viscosity Solutions}, volume~25 of {\em Stochastic Modeling and Applied Probability}.
\newblock Springer-Verlag, 2nd edition.

\bibitem[Hata, 2017]{hataRisksensitiveAssetManagement2017a}
Hata, H. (2017).
\newblock Risk-sensitive asset management in a general diffusion factor model: Risk-seeking case.
\newblock {\em Japan Journal of Industrial and Applied Mathematics}, 34(1), 59--98.

\bibitem[Hata, 2018]{hataRisksensitivePortfolioOptimization2018}
Hata, H. (2018).
\newblock Risk-sensitive portfolio optimization problem for a large trader with inside information.
\newblock {\em Japan Journal of Industrial and Applied Mathematics}, 35(3), 1037--1063.

\bibitem[Hata, 2021]{hataRiskSensitiveAssetManagement2021}
Hata, H. (2021).
\newblock Risk-{{Sensitive Asset Management}} with {{Lognormal Interest Rates}}.
\newblock {\em Asia-Pacific Financial Markets}, 28(2), 169--206.

\bibitem[Hata \& Iida, 2006]{hataRisksensitiveStochasticControl2006}
Hata, H. \& Iida, Y. (2006).
\newblock A risk-sensitive stochastic control approach to an optimal investment problem with partial information.
\newblock {\em Finance and Stochastics}, 10(3), 395--426.

\bibitem[Kuroda \& Nagai, 2002]{kuna02}
Kuroda, K. \& Nagai, H. (2002).
\newblock Risk-sensitive portfolio optimization on infinite time horizon.
\newblock {\em Stochastics and Stochastics Reports}, 73, 309--331.

\bibitem[Liptser \& Shiryaev, 2004]{LiShI04}
Liptser, R. \& Shiryaev, A. (2004).
\newblock {\em Statistics of Random Processes: I. General Theory}.
\newblock Probability and Its Applications. Springer-Verlag, 2 edition.

\bibitem[Lleo \& MacLean, 2025]{lleoDualDominanceHow2025}
Lleo, S. \& MacLean, L.~C. (2025).
\newblock Dual dominance: how {Harry} {Markowitz} and {William} {Ziemba} impacted portfolio management.
\newblock {\em Annals of Operations Research}, 346(1), 181--216.

\bibitem[Lleo \& Runggaldier, 2024]{LleoRunggaldierSeparation24}
Lleo, S. \& Runggaldier, W. (2024).
\newblock On the separation of estimation and control in risk-sensitive investment problems under incomplete observation.
\newblock {\em European Journal of Operational Research}, 316(1), 200--214.

\bibitem[Lleo \& Runggaldier, 2026]{LleoRunggaldier_EntropyRegularizationinRLandRSIM_2026}
Lleo, S. \& Runggaldier, W. (2026).
\newblock Reinforcement learning for risk-sensitive investment management: A free energy--entropy duality approach.
\newblock Preprint.

\bibitem[Meneghini, 1994]{Meneghini1994}
Meneghini, L. (1994).
\newblock {\em Modelli Risolvibili per Problemi di Controllo di Sistemi Dinamici Imprecisi Multivariati}.
\newblock PhD thesis, Universit\`{a} Degli Studi Di Padova.

\bibitem[Nagai \& Peng, 2002]{nape02}
Nagai, H. \& Peng, S. (2002).
\newblock Risk-sensitive dynamic portfolio optimization with partial information on infinite time horizon.
\newblock {\em The Annals of Applied Probability}, 12(1), 173--195.

\end{thebibliography}

\appendix

\section{Proofs}

\subsection{Recalling Lemma 5.3.1 from \citet{Meneghini1994} along with its proof}\label{app:Lemma:Meneghini}

\begin{lemma}\label{lem:infexp:app}
 Let $f: X\,\rightarrow\,{\mathbb{R}^+}$ be a function such that $\inf_{x\in X}f(x)>0$. Then
 $$\log\left(\inf_{x\in X}f(x)\right)=\inf_{x\in X}\log(f(x))$$
\end{lemma}

\begin{proof}
{}
\begin{itemize}\item[i)] First we prove the one-sided inequality $(\le)$

For any $\bar x\in X$ we have $\inf_{x\in X}f(x)\le f(\bar x)$ which, by the monotonicity of $\log(\cdot)$, implies
$$\log\left(\inf_{x\in X}f(x)\right)\le \log\left(f(\bar x)\right)\quad\mbox{and, with it,}\quad \log\left(\inf_{x\in X}f(x)\right)\le \inf_{x\in X}\log\left(f(x)\right)$$
\item[ii)] Inverse inequality $(\ge)$

By the definition of $\inf$, for each $n\in\mathbb{N}$ there exists $x_n\in X$ such that
$$\inf_{x\in X}f(x)\le f(x_n)\le \inf_{x\in X}f(x)+\frac{1}{n}$$
By the monotonicity of $\log(\cdot)$ it then follows
$$\log f(x_n)\le \log\left( \inf_{x\in X}f(x)+\frac{1}{n}\right)\quad\mbox{and, with it,}\quad \inf_{x\in X}\log f(x)\le \log\left( \inf_{x\in X}f(x)+\frac{1}{n}\right)$$
Passing to the limit when $n\to\infty$ and thereby using the continuity of $\log(\cdot)$, 
$$\inf_{x\in X}\log f(x)\le \log\left( \inf_{x\in X}f(x)\right)$$
\end{itemize}
\end{proof}

\subsection{Proof of Proposition \ref{prop:candidate_controls}}\label{app:proof:prop:candidate_controls}

We express the Bellman--Isaacs equation \eqref{eq:BellmanIsaacs:u} as 
\begin{align}\label{eq:BellmanIsaacs:u:2}
    & \frac{\partial u(s,x)}{\partial s} + \inf_{h \in \mathbb{R}^m} \sup_{\gamma \in \mathbb{R}^d} \Bigg\{
    \left[b_s + B_s x + \Lambda_s \gamma \right]' Du(s,x) 
    + \frac{1}{2} \tr \left(\Lambda_s\Lambda_s' D^2u(s,x) \right)
    +\theta g(s,x,h,\gamma;\theta) \Bigg\}
                        \nonumber\\
    =& \frac{\partial u(s,x)}{\partial s}
        + \left[ \left(b_s + B_s x \right)' Du(s,x) 
        + \frac{1}{2} \tr \left(\Lambda_s\Lambda_s' D^2u(s,x) \right)
        + \theta \left(c_s + C_s x \right)
        - \frac{\theta}{2}\Xi_s'\Xi_s 
        \right]
                        \nonumber\\
    &+ \inf_{h \in \mathbb{R}^m} \sup_{\gamma \in \mathbb{R}^d} \theta F(s,x,h,\gamma;\theta,Du(s,x)),
\end{align}

where
\begin{align}\label{eq:proof:candidate_controls:F}
    F(s,x,h,\gamma;\theta,p)
    := \frac{1}{2} h'\Sigma_s\Sigma_s'h 
        - h' (a_s + A_s x) 
        - \frac{1}{2\theta} \gamma'\gamma
        - \gamma'\left(\Sigma_s'h - \Xi_s \right)
        + \frac{1}{\theta}\gamma'\Lambda_s' p.
\end{align}

To derive the two equivalent representations of the candidate policies, we optimize in the two possible orders over $h$ and $\gamma$.

For the first representation, \eqref{eq:hhat}--\eqref{eq:gammahat:1}, we first solve the supremum problem in $\gamma$. The function $F$ is quadratic in $\gamma$. Its unique maximizer corresponds to the candidate policy
\begin{align}\label{eq:gammahat:proof:1}
    \hat{\gamma}(s,x,h,Du(s,x))
    = \Lambda_s' Du(s,x) - \theta \left(\Sigma_s'h - \Xi_s \right).
\end{align}
Substituting into $F$, we get 

\begin{align}\label{eq:proof:candidate_controls:F:2}
    & F\left(s,x,h,\hat{\gamma}(s,x,h,Du(s,x));\theta,Du(s,x)\right)
                    \nonumber\\
    =&  \frac{\theta+ 1}{2} h'\Sigma_s\Sigma_s'h 
        - h' \left( a_s + A_s x + \Sigma_s\Lambda_s'Du(s,x)+ \theta \Sigma_s\Xi_s\right)
        + \frac{\theta}{2}\Xi_s'\Xi_s
        + \Xi_s'\Lambda_s'Du(s,x)
                     \nonumber\\  
    &   + \frac{1}{2\theta}     
            Du(s,x)'\Lambda_s\Lambda_s'Du(s,x).
                    \nonumber\\  
    =& \frac{\theta+ 1}{2} \left[
        h - 
        \frac{1}{\theta+1} \left(\Sigma_s\Sigma_s'\right)^{-1}
            \left( a_s + A_s x
            + \theta \Sigma_s\Xi_s + \Sigma_s\Lambda_s' Du(s,x)
        \right)
    \right]' \left(\Sigma_s\Sigma_s'\right)
                        \nonumber\\  
    & \cdot \left[
        h - 
        \frac{1}{\theta+1} \left(\Sigma_s\Sigma_s'\right)^{-1}
            \left( a_s + A_s x
            + \theta \Sigma_s\Xi_s + \Sigma_s\Lambda_s' Du(s,x)
        \right)
    \right]
                        \nonumber\\  
    &   - \frac{1}{2(\theta+ 1)} \left( 
            a_s + A_s x
            + \theta \Sigma_s\Xi_s \right)' \left(\Sigma_s\Sigma_s'\right)^{-1}\left(
            a_s + A_s x
            + \theta \Sigma_s\Xi_s
        \right)
                        \nonumber\\
    &   - \frac{1}{\theta+ 1} \left( 
            a_s + A_s x
            + \theta \Sigma_s\Xi_s \right)'\left(\Sigma_s\Sigma_s'\right)^{-1}\Sigma_s\Lambda_s' Du(s,x)
                        \nonumber\\
    &   + \frac{\theta}{2}\Xi_s'\Xi_s
        + \Xi_s'\Lambda_s'Du(s,x)
        + \frac{1}{2\theta} Du'(s,x)\Lambda_s \mathcal{P}_s^{-}(\theta) \Lambda_s' Du(s,x)
                        \nonumber\\ 
\end{align}

Next, we solve the infimum problem in $h$, given the maximizer in $\gamma$ obtained above. The function $F$ is quadratic in $h$. Its unique minimizer corresponds to the candidate policy
\begin{align}\label{eq:hhat:proof}
\hat{h}(s,x,Du(s,x))
    =  \frac{1}{\theta+1}\left(\Sigma_s\Sigma_s'\right)^{-1}
\left( a_s + A_s x
+ \theta \Sigma_s\Xi_s + \Sigma_s\Lambda_s' Du(s,x)
\right).
\end{align}
Hence, at the minimizer $\hat h$, $F$ takes the value
\begin{align}\label{eq:proof:candidate_controls:F:min:h} 
& F\left(s,x,\hat h(s,x,Du(s,x)),\hat \gamma(s,x,Du(s,x));\theta,Du(s,x)\right) 
                        \nonumber\\ 
=& - \frac{1}{2(\theta+ 1)} \left( a_s + A_s x + \theta \Sigma_s\Xi_s \right)' \left(\Sigma_s\Sigma_s'\right)^{-1}\left( a_s + A_s x + \theta \Sigma_s\Xi_s \right) 
                        \nonumber\\ 
    & - \frac{1}{\theta+ 1} \left( a_s + A_s x + \theta \Sigma_s\Xi_s \right)'\left(\Sigma_s\Sigma_s'\right)^{-1}\Sigma_s\Lambda_s' Du(s,x) 
                        \nonumber\\ 
    &   + \frac{\theta}{2}\Xi_s'\Xi_s
        + \Xi_s'\Lambda_s'Du(s,x)
        + \frac{1}{2\theta} Du'(s,x)\Lambda_s \mathcal{P}_s^{-}(\theta) \Lambda_s' Du(s,x)
\end{align}

Finally, we substitute \eqref{eq:hhat:proof} into \eqref{eq:gammahat:proof:1} to get
\begin{align}\label{eq:gammahat:proof:2}
    \hat{\gamma}(s,x, Du(s,x)) 
    =&  \Lambda_s' Du(s,x) - \theta \left[\frac{1}{\theta+1} \Sigma_s'\left(\Sigma_s\Sigma_s'\right)^{-1}
        \left( a_s + A_s x + \theta \Sigma_s\Xi_s + \Sigma_s\Lambda_s' Du(s,x) \right) - \Xi_s \right]
                                    \nonumber\\
    =&  \mathcal{P}_s^{-}(\theta)\Lambda_s' Du(s,x) 
    - \frac{\theta}{\theta+1} \Sigma_s'\left(\Sigma_s\Sigma_s'\right)^{-1}
        \left( a_s + A_s x \right) 
    + \theta \mathcal{P}_s^{-}(\theta) \Xi_s. 
\end{align}

Next, we reverse the order of optimization and begin with the infimum problem in $h$. The function $F$ is quadratic in $h$. Its unique minimizer corresponds to the candidate policy
\begin{align}\label{eq:hhat:proof:alt:1}
    \hat h(s,x,\gamma) 
    &= \left(\Sigma_s\Sigma_s'\right)^{-1}  
        \left[ (a_s + A_s x) + \Sigma_s \gamma\right].
\end{align}

Substituting into $F$, we get
\begin{align}\label{eq:proof:candidate_controls:F:alt:2}
    & F\left(s,x,\hat h(s,x,\gamma),\gamma;\theta,Du(s,x)\right)
                        \nonumber\\
    =& \frac{1}{2}  
        \left[ (a_s + A_s x) + \Sigma_s \gamma\right]'\underbrace{\left(\Sigma_s\Sigma_s'\right)^{-1}\Sigma_s\Sigma_s'}_{I_m}\left(\Sigma_s\Sigma_s'\right)^{-1}  
                        \nonumber\\
    &   \cdot \left[ (a_s + A_s x) + \Sigma_s \gamma\right] 
        - (a_s + A_s x)'\left(\Sigma_s\Sigma_s'\right)^{-1}  
        \left[ (a_s + A_s x) + \Sigma_s \gamma\right] 
                        \nonumber\\
    &   - \frac{1}{2\theta} \gamma'\gamma
        - \gamma'\Sigma_s'\left(\Sigma_s\Sigma_s'\right)^{-1}  
        \left[ (a_s + A_s x) + \Sigma_s \gamma\right]
         + \gamma'\Xi_s
         + \frac{1}{\theta}\gamma'\Lambda_s' Du(s,x)
                        \nonumber\\
    =& - \frac{1}{2\theta} \left\{ \gamma -\left[
            \mathcal{P}_s^{+}(\theta)
        \right]^{-1} 
        \left[ 
            - \theta \Sigma_s' \left(\Sigma_s\Sigma_s'\right)^{-1} (a_s + A_s x)
            + \theta \Xi_s
            + \Lambda_s'Du(s,x)
        \right]\right\}'
                        \nonumber\\
    & 
        \mathcal{P}_s^{+}(\theta)
        \left\{ \gamma -\left[
            \mathcal{P}_s^{+}(\theta)
        \right]^{-1} 
        \left[ 
            - \theta \Sigma_s' \left(\Sigma_s\Sigma_s'\right)^{-1} (a_s + A_s x)
            + \theta \Xi_s
            + \Lambda_s'Du(s,x)
        \right]\right\}
                        \nonumber\\
    &   + \frac{1}{2\theta}  
        \left[ 
            - \theta \Sigma_s' \left(\Sigma_s\Sigma_s'\right)^{-1} (a_s + A_s x)
            + \theta \Xi_s
            + \Lambda_s'Du(s,x)
        \right]'
                        \nonumber\\
    & 
        \left[
            \mathcal{P}_s^{+}(\theta)
        \right]^{-1}  
        \left[ 
            - \theta \Sigma_s' \left(\Sigma_s\Sigma_s'\right)^{-1} (a_s + A_s x)
            + \theta \Xi_s
            + \Lambda_s'Du(s,x)
        \right]
        - \frac{1}{2} (a_s + A_s x)' \left(\Sigma_s\Sigma_s'\right)^{-1}(a_s + A_s x)
\end{align}

We then solve the supremum for $\gamma$, given the value we found for $\hat{h}$. The function $F$ is quadratic and strictly concave in $\gamma$. Its unique maximizer corresponds to the candidate policy
\begin{align}\label{eq:gammahat:proof}
\hat \gamma (s,x,Du(s,x))
    =  \left[
            \mathcal{P}_s^{+}(\theta)
        \right]^{-1} 
        \left[ 
            - \theta \Sigma_s' \left(\Sigma_s\Sigma_s'\right)^{-1} (a_s + A_s x)
            + \theta \Xi_s
            + \Lambda_s'Du(s,x)
        \right]
\end{align}

Substituting the maximizer \eqref{eq:gammahat:proof} into \eqref{eq:hhat:proof:alt:1} recovers \eqref{eq:hhat:alt}, and hence, at the maximizer $\hat \gamma$, $F$ takes the value
\begin{align}\label{eq:proof:candidate_controls:F:max:gamma}
    & F\left(s,x,\hat h(s,x,Du(s,x)),\hat \gamma(s,x,Du(s,x));\theta,Du(s,x)\right)
                        \nonumber\\
    =& \frac{1}{2\theta}  
        \left[ 
            - \theta \Sigma_s' \left(\Sigma_s\Sigma_s'\right)^{-1} (a_s + A_s x)
            + \theta \Xi_s
            + \Lambda_s'Du(s,x)
        \right]'
        \left[
            \mathcal{P}_s^{+}(\theta)
        \right]^{-1}  
                        \nonumber\\
    &   \cdot \left[ 
            - \theta \Sigma_s' \left(\Sigma_s\Sigma_s'\right)^{-1} (a_s + A_s x)
            + \theta \Xi_s
            + \Lambda_s'Du(s,x)
        \right]
       - \frac{1}{2} (a_s + A_s x)' \left(\Sigma_s\Sigma_s'\right)^{-1}(a_s + A_s x)
                        \nonumber\\ 
\end{align}
The final step is to apply \eqref{eq:projinverse} to \eqref{eq:proof:candidate_controls:F:max:gamma} and to develop the expression to obtain \eqref{eq:proof:candidate_controls:F:min:h}. 

Thus, the same saddle value is obtained under either order of optimization. Since $F$ at \eqref{eq:proof:candidate_controls:F:min:h} and \eqref{eq:proof:candidate_controls:F:max:gamma} are equal, we conclude that $\mathcal{H}^+$ given at \eqref{eq:Hamiltonian:H+} and $\mathcal{H}^-$ given at \eqref{eq:Hamiltonian:H-} are equal,  $\mathcal{H}^+(s,x, p, M) = \mathcal{H}^-(s,x, p, M)$, so the stochastic differential game is well defined.

\subsection{Proof of Theorem \ref{theo:sol_IsaacPDE}}\label{app:proof:theo:sol_IsaacPDE}

With \eqref{eq:proof:candidate_controls:F:min:h}, the Bellman--Isaacs equation at \eqref{eq:BellmanIsaacs:u:2} becomes
\begin{align}
    &   \frac{\partial u(s,x)}{\partial s}
        + \left[ \left(b_s + B_s x \right)' Du(s,x) 
        + \frac{1}{2} \tr \left(\Lambda_s\Lambda_s' D^2u(s,x) \right)
        + \theta \left(c_s + C_s x \right)
        - \frac{\theta}{2}\Xi_s'\Xi_s 
        \right]   
                        \nonumber\\
    &   + \theta \Bigg[ - \frac{1}{2(\theta+ 1)} \left( a_s + A_s x + \theta \Sigma_s\Xi_s \right)' \left(\Sigma_s\Sigma_s'\right)^{-1}\left( a_s + A_s x + \theta \Sigma_s\Xi_s \right) 
                        \nonumber\\ 
    & - \frac{1}{\theta+ 1} \left( a_s + A_s x + \theta \Sigma_s\Xi_s \right)'\left(\Sigma_s\Sigma_s'\right)^{-1}\Sigma_s\Lambda_s' Du(s,x) 
                        \nonumber\\ 
    &   + \frac{\theta}{2}\Xi_s'\Xi_s
        + \Xi_s'\Lambda_s'Du(s,x)
        + \frac{1}{2\theta} Du'(s,x)\Lambda_s \mathcal{P}_s^{-}(\theta) \Lambda_s' Du(s,x) \Bigg]
    = 0.
\end{align}
Next, substitute the derivatives of the ansatz \eqref{eq:sol:u}, namely
\[
Du(s,x)=-\theta(Q_sx+q_s), \qquad D^2u(s,x)=-\theta Q_s,
\]
to obtain
\begin{align}
    &   -\theta\left(\frac{1}{2}x'\dot{Q}_s x + x' \dot{q}_s + \dot{k}_s\right)
        + \left[ -\theta \left(b_s + B_s x \right)' \left(Q_s x + q_s\right) 
        - \frac{\theta}{2} \tr \left(\Lambda_s\Lambda_s' Q_s \right)
        + \theta \left(c_s + C_s x \right)
        - \frac{\theta}{2}\Xi_s'\Xi_s 
        \right]   
                        \nonumber\\
    &   + \theta \Bigg[ - \frac{1}{2(\theta+ 1)} \left( a_s + A_s x + \theta \Sigma_s\Xi_s \right)' \left(\Sigma_s\Sigma_s'\right)^{-1}\left( a_s + A_s x + \theta \Sigma_s\Xi_s \right) 
                        \nonumber\\ 
    & + \frac{\theta}{\theta+ 1} \left( a_s + A_s x + \theta \Sigma_s\Xi_s \right)'\left(\Sigma_s\Sigma_s'\right)^{-1}\Sigma_s\Lambda_s' (Q_sx+q_s) 
                        \nonumber\\ 
    &   + \frac{\theta}{2}\Xi_s'\Xi_s
        - \theta \Xi_s'\Lambda_s'(Q_sx+q_s)
        + \frac{\theta}{2} (Q_sx+q_s)'\Lambda_s \mathcal{P}_s^{-}(\theta) \Lambda_s' (Q_sx+q_s) \Bigg]
    = 0.
\end{align}
where we used the notation $\dot{Q}_s = \frac{dQ}{ds}, \dot{q}_s = \frac{dq}{ds}, \dot{k}_s = \frac{dk}{ds}$.

Divide by $-\theta$, rearrange, and factor into a quadratic form in $x$:
\begin{align}\label{eq:proof:sol_IsaacPDE:interm_final}
    &   \frac{1}{2} x' \Bigg\{ 
            \dot{Q}_s
            + 2 B_s'Q_s
            + \frac{1}{(\theta+1)} A_s' (\Sigma_s\Sigma_s')^{-1} A_s
            - \frac{2\theta}{(\theta+1)} A_s'(\Sigma_s\Sigma_s')^{-1} \Sigma_s\Lambda_s'Q_s
            - \theta  Q_s'\Lambda_s \mathcal{P}_s^{-}(\theta) \Lambda_s'Q_s
        \Bigg\} x
                                        \nonumber\\
    &   + x' \Bigg\{
            \dot{q}_s
            + Q_s' b_s + B_s'q_s
            - C_s'
            + \frac{1}{\theta+1} A_s'(\Sigma_s\Sigma_s')^{-1}
        \left( a_s + \theta \Sigma_s\Xi_s \right)
            - \frac{\theta}{\theta+1} 
         A_s'(\Sigma_s\Sigma_s')^{-1}\Sigma_s\Lambda_s' q_s
                                        \nonumber\\
        &  - \theta Q_s' \Lambda_s\mathcal{P}_s^{-}(\theta)\Lambda_s' q_s
            - \frac{\theta}{\theta+1}  Q_s'\Lambda_s\Sigma_s'(\Sigma_s\Sigma_s')^{-1}
        \left( a_s + \theta \Sigma_s\Xi_s \right)
            + \theta Q_s'\Lambda_s\Xi_s 
        \Bigg\}
                                        \nonumber\\
      & + \Bigg\{
            \dot{k}_s
            + b_s'q_s 
            - \frac{\theta}{2} q_s' \Lambda_s\Lambda_s' q_s 
            + \theta \Xi_s'\Lambda_s'q_s
            + \frac{1}{2} \tr \left(\Lambda_s\Lambda_s' Q_s \right)
            - c_s
            - \frac{\theta-1}{2}\Xi_s'\Xi_s 
                                        \nonumber\\
    &       + \frac{1}{2(\theta+1)} \left(a_s + \theta \Sigma_s\Xi_s\right)'(\Sigma_s\Sigma_s')^{-1} \left(a_s + \theta \Sigma_s\Xi_s\right)  
            + \frac{\theta^2}{2(\theta+1)}  q_s'\Lambda_s\Sigma_s'(\Sigma_s\Sigma_s')^{-1}
    \Sigma_s\Lambda_s'q_s
                                        \nonumber\\
    &       - \frac{\theta}{\theta+1} \left( a_s + \theta \Sigma_s\Xi_s \right)'(\Sigma_s\Sigma_s')^{-1}\Sigma_s\Lambda_s'q_s
        \Bigg\}
    = 0.
\end{align}

Equation \eqref{eq:proof:sol_IsaacPDE:interm_final} holds if and only if the coefficients of the quadratic term in $x$, of the linear term in $x$, and the deterministic term all equal 0. We now address each of these terms to express them into a more convenient form, starting with the quadratic term in $x$: 
\begin{align}
    & \dot{Q}_s
    + 2 B_s'Q_s
    + \frac{1}{(\theta+1)} A_s' (\Sigma_s\Sigma_s')^{-1} A_s
    - \frac{2\theta}{(\theta+1)} A_s'(\Sigma_s\Sigma_s')^{-1} \Sigma_s\Lambda_s'Q_s
    - \theta  Q_s'\Lambda_s \mathcal{P}_s^{-}(\theta) \Lambda_s'Q_s
    = 0
                                            \nonumber\\
    \Leftrightarrow 
    &\dot{Q}_s
    - \theta  Q_s'\Lambda_s \mathcal{P}_s^{-}(\theta) \Lambda_s'Q_s
    + \left(B_s' - \frac{\theta}{(\theta+1)} A_s'(\Sigma_s\Sigma_s')^{-1} \Sigma_s\Lambda_s' 
    \right)Q_s 
                                            \nonumber\\
    &+ Q_s'\left(B_s - \frac{\theta}{(\theta+1)} \Lambda_s\Sigma_s'(\Sigma_s\Sigma_s')^{-1}A_s  
    \right)
    + \frac{1}{(\theta+1)} A_s' (\Sigma_s\Sigma_s')^{-1} A_s
    = 0.
\end{align}     
This final expression is a Riccati equation for the function $Q_s$. 

Moving on to the linear term in $x$, we get 
\begin{align}
    & \dot{q}_s
        + Q_s' b_s + B_s'q_s
        - C_s'
        + \frac{1}{\theta+1} A_s'(\Sigma_s\Sigma_s')^{-1}
    \left( a_s + \theta \Sigma_s\Xi_s \right)
        - \frac{\theta}{\theta+1} 
         A_s'(\Sigma_s\Sigma_s')^{-1}\Sigma_s\Lambda_s' q_s
                                        \nonumber\\
    &  - \theta Q_s' \Lambda_s\mathcal{P}_s^{-}(\theta)\Lambda_s' q_s
    - \frac{\theta}{\theta+1}  Q_s'\Lambda_s\Sigma_s'(\Sigma_s\Sigma_s')^{-1}
    \left( a_s + \theta \Sigma_s\Xi_s \right)
        + \theta Q_s'\Lambda_s\Xi_s 
       = 0 
                                        \nonumber\\
    \Leftrightarrow
    & \dot{q}_s
        + \left(B_s'- \frac{\theta}{\theta+1} A_s'(\Sigma_s\Sigma_s')^{-1}\Sigma_s\Lambda_s' \right)q_s
        - \theta Q_s' \Lambda_s\mathcal{P}_s^{-}(\theta)\Lambda_s' q_s
                                        \nonumber\\
    &   + Q_s' \left[b_s - \frac{\theta}{\theta+1}  \Lambda_s\Sigma_s'(\Sigma_s\Sigma_s')^{-1}
    \left( a_s + \theta \Sigma_s\Xi_s \right) + \theta \Lambda_s\Xi_s\right]
                                        \nonumber\\
    &   - C_s'
        + \frac{1}{\theta+1} A_s'(\Sigma_s\Sigma_s')^{-1}
    \left( a_s + \theta \Sigma_s\Xi_s \right)
        = 0,
\end{align}
which is a linear ODE for $q_s$.

Finally, 
\begin{align}
    k_t  
    =& \int_t^T \Bigg\{ 
        \frac{\theta}{2} q_s'\Lambda_s \mathcal{P}_s^{-}(\theta)\Lambda_s' q_s
       - \left[b_s'+ \theta \Xi_s'\Lambda_s' - \frac{\theta}{\theta+1} \left( a_s + \theta \Sigma_s\Xi_s \right)'(\Sigma_s\Sigma_s')^{-1}\Sigma_s\Lambda_s' \right]q_s 
                                        \nonumber\\
    &   - \frac{1}{2} \tr \left(\Lambda_s\Lambda_s' Q_s \right)
        - \frac{1}{2(\theta+1)} \left(a_s + \theta \Sigma_s\Xi_s\right)'(\Sigma_s\Sigma_s')^{-1} \left(a_s + \theta \Sigma_s\Xi_s\right)
        + c_s
        + \frac{\theta-1}{2}\Xi_s'\Xi_s \Bigg\} ds.      
\end{align}

This yields \eqref{eq:Q:Riccati}, \eqref{eq:q:ODE}, and \eqref{eq:k:integral}, and hence proves the theorem.

\subsection{Proof of Theorem \ref{theo:verification}}\label{app:proof:theo:verif}

1. We start by establishing that the candidate strategies $H^*$ and $\Gamma^*$ from Proposition \ref{prop:candidate_controls} are Borel-measurable and Markov.
For the value function $u(t,x)$ in Theorem \ref{theo:sol_IsaacPDE}, the candidate strategies $H^*$ and $\Gamma^*$ induced by \eqref{eq:hhat} and \eqref{eq:gammahat:2} simplify to $H^* = \left(h^*\left(s,X_s\right) \right)_{s\in [t,T]}$ and $\Gamma^* = \left(\gamma^*\left(s,X_s\right) \right)_{s\in [t,T]}$  with

\begin{align}
    h^*(s,x)
    =&  
    \frac{1}{\theta+1}\left(\Sigma_s\Sigma_s'\right)^{-1}\left[
        \left(A_s - \theta \Sigma_s\Lambda_s' Q_s \right)x
        + a_s
        + \theta \Sigma_s (\Xi_s - \Lambda_s' q_s)
        \right]
                                \label{eq:hhat:verif}\\
    \gamma^*(s,x) 
    =&  
    \left\{
        - \frac{\theta}{\theta+1} \Sigma_s'\left(\Sigma_s\Sigma_s'\right)^{-1}A_s 
        - \theta \mathcal{P}_s^{-}(\theta)\Lambda_s' Q_s
    \right\}x 
    - \frac{\theta}{\theta+1} \Sigma_s'\left(\Sigma_s\Sigma_s'\right)^{-1}a_s
     + \theta \mathcal{P}_s^{-}(\theta) ( \Xi_s- \Lambda_s' q_s).
                                \label{eq:gammahat:verif}
\end{align}  
    
Therefore, the maps defining $h^*(s,X_s)$ at \eqref{eq:hhat:verif} and $\gamma^*(s,X_s)$ at \eqref{eq:gammahat:verif} are Borel-measurable. Therefore, the defining policies $h^*(\cdot)$ and $\gamma^*(\cdot)$ are Borel-measurable, and the induced strategies $H^*$ and $\Gamma^*$ are Markov.

Moreover, under $\mathbb{P}^{\Gamma^*}$, the state process $X = \left(X_s\right)_{s \in [t,T]}$ remains Gaussian. To see this, substitute \eqref{eq:gammahat:verif} into \eqref{eq:state:Pgamma:FO} to obtain the following dynamics:
\begin{align}\label{eq:state:Pgamma:FO:verif}
    d X_s
    =& \Bigg\{ \left[
        b_s 
        - \frac{\theta}{\theta+1} \Lambda_s\Sigma_s'\left(\Sigma_s\Sigma_s'\right)^{-1}a_s
     + \theta \Lambda_s\mathcal{P}_s^{-}(\theta) ( \Xi_s- \Lambda_s' q_s)
        \right]
                            \nonumber\\
    &+ \left[ 
        B_s 
         - \frac{\theta}{\theta+1} \Lambda_s \Sigma_s'\left(\Sigma_s\Sigma_s'\right)^{-1}A_s 
        - \theta \Lambda_s\mathcal{P}_s^{-}(\theta)\Lambda_s' Q_s
        \right] X_s 
        \Bigg\}ds
        + \Lambda_s dW^{\Gamma^*}_s,
\end{align}

From \eqref{eq:hhat:verif} and \eqref{eq:state:Pgamma:FO:verif}, by which $h^*(s,x)$ is affine in $x$ and $X$ is a Gaussian process, we can conclude that the strategy $H^*$ is a Markov strategy in class $\mathcal{A}^H$. Since $h^*(s,X_s)$ is affine in the Gaussian state $X_s$, it is square-integrable on $[t,T]$. Hence $H^*\in\mathcal A^H$. To show that $\Gamma^*$ is Markov and in class
$\mathcal{A}^{\Gamma}$ recall that, by Theorem \ref{theo:sol_IsaacPDE}, $u(t,x)$ is quadratic in $x$ and so 
$\hat\gamma(s,x,Du(s,x))$ in \eqref{eq:gammahat:1} of Proposition \ref{prop:candidate_controls} is affine in $x$. Moreover, under $\mathbb{P}^{\Gamma^*}$, the state process $X$ given by \eqref{eq:state:Pgamma:FO:verif} is Gaussian.

For convenience, define the density increment process associated with $\Gamma^*$ on $[t,T]$ by
\begin{align}\label{eq:chi:GammaStar:ts}
\chi_{t,s}^{\Gamma^*}
:=
\exp\left\{
-\frac12\int_t^s \|\gamma^*(u,X_u)\|^2\,du
+\int_t^s \gamma^*(u,X_u)' \, dW_u
\right\},
\qquad s\in[t,T].
\end{align}
Equivalently, $\chi_{t,s}^{\Gamma^*}=\chi_s^{\Gamma^*}/\chi_t^{\Gamma^*}$, where $\chi^{\Gamma^*}$ is defined in \eqref{eq:RNderivative:gamma}.

Since $\gamma^*(s,X_s)$ is affine in the Gaussian state $X_s$, it is square-integrable on $[t,T]$. Under the standing $C^1$ assumptions on the deterministic coefficients and the finite time horizon, the coefficients are bounded and the Novikov condition is satisfied. Therefore, the process $\big(\chi_{t,s}^{\Gamma^*}\big)_{s\in[t,T]}$ is a true martingale. 
Hence, condition (iii) in Definition \ref{def:classAgammaT:fullobs} holds. Since $\gamma^*(s,x)$ is affine in $x$, conditions (i) and (ii) in Definition \ref{def:classAgammaT:fullobs} also hold. Therefore, $\Gamma^* \in \mathcal{A}^{\Gamma}$.
\medskip

2. We focus on the first inequality. For an admissible strategy $\Gamma = \left( \gamma_s\right)_{s \in [t,T]} \in \mathcal{A}^\Gamma$, 
Remark \ref{rk:candidate:saddle} implies that 
\begin{align}
        \frac{\partial u(s,x)}{\partial s} 
    + \left[b_s + B_s x + \Lambda_s \gamma_s \right]' Du(s,x) 
    + \frac{1}{2} \tr \left(\Lambda_s\Lambda_s' D^2u(s,x) \right)
    +\theta g\left(s,x,h^*(s,x),\gamma_s;\theta\right)
    \leq 0.
\end{align}
By Dynkin's formula,
\begin{align}
    & \mathbf{E}_{t,x}^{\mathbb{P}^\Gamma} \left[ u(T,X_T)\right]
                                    \nonumber\\
    =& u(t,x) + \mathbf{E}_{t,x}^{\mathbb{P}^\Gamma} \left[ \int_t^T \left\{       \frac{\partial u(s,X_s)}{\partial s} 
        + \left[b_s + B_s X_s + \Lambda_s \gamma_s \right]' Du(s,X_s) 
        + \frac{1}{2} \tr \left(\Lambda_s\Lambda_s' D^2u(s,X_s) \right)
        \right\} ds \right]
                                    \nonumber\\
    \leq& u(t,x) - \mathbf{E}_{t,x}^{\mathbb{P}^\Gamma} \left[ \int_t^T \theta g\left(s,X_s,h^*(s,X_s),\gamma_s;\theta\right) ds \right]
\end{align}
Recall that $u(T,y) = 0 \; \forall y \in \mathbb{R}^n$, hence
\begin{align}
    u(t,x) \geq \mathbf{E}_{t,x}^{\mathbb{P}^\Gamma} \left[ \int_t^T \theta g\left(s,X_s,h^*(s,X_s),\gamma_s;\theta\right) ds \right].
\end{align}
If $\Gamma = \Gamma^*$ then equality holds, that is,
\begin{align}
    u(t,x) = \mathbf{E}_{t,x}^{\mathbb{P}^{\Gamma^*}} \left[ \int_t^T \theta g\left(s,X_s,h^*(s,X_s),\gamma^*(s,X_s);\theta\right) ds \right].
\end{align}
Similarly, for any admissible strategy $H=(h_s)_{s\in[t,T]}\in\mathcal A^H$, the local saddle-point condition implies
\begin{align}
\frac{\partial u}{\partial s}(s,x)
+ \left[b_s+B_sx+\Lambda_s\gamma^*(s,x)\right]'Du(s,x)
+ \frac12\tr\!\left(\Lambda_s\Lambda_s'D^2u(s,x)\right)
+ \theta g\left(s,x,h_s,\gamma^*(s,x);\theta\right)\ge 0.
\end{align}
Applying Dynkin’s formula under $\mathbb P^{\Gamma^*}$ yields the second inequality in \eqref{eq:verification:saddle}.
\medskip

3. This result follows directly from 1. and 2. above.

\end{document}